\documentclass[a4paper,12pt]{article}

\usepackage{amsmath,amsthm,amssymb,amsfonts}

\makeatletter

\@addtoreset{equation}{section}
\makeatother

\newtheorem{theorem}{Theorem}[section]
\newtheorem{lemma}[theorem]{Lemma}

\def\ep{\varepsilon}

\def\R{\mathbb R}
\def\S{\mathbb S}
\def\N{\mathbb N}
\def\pa{\partial}
\def\b{\backslash}
\def\A{\mathbf A}

\def\V{\mathcal V}
\def\B{\mathcal B}
\def\A{\mathcal A}
\def\G{\mathcal G}

\begin{document}


\title{Inverse scattering at high energies for the multidimensional Newton equation in a long range potential}
\author{Alexandre Jollivet}

\maketitle

\begin{abstract} 
We define scattering data for the Newton equation in a potential $V\in C^2(\R^n,\R)$, $n\ge2$, that decays at infinity like $r^{-\alpha}$ for some $\alpha\in (0,1]$. We provide estimates on the scattering solutions and scattering data and we prove, in particular, that the scattering data at high energies
uniquely determine the short range part of the potential up to the knowledge of the long range tail of the potential. The Born approximation at fixed energy of the scattering data is also considered.
We then change the definition of the scattering data to study inverse scattering in other asymptotic regimes.
These results were obtained by developing the inverse scattering approach of [Novikov, 1999].  
\end{abstract}


\section{Introduction}
\label{int}
Consider the multidimensional Newton equation in an external static force $F$ deriving from a scalar potential $V$:
\begin{equation}
\ddot x(t)=F(x(t))=-\nabla V(x(t)),\label{1.1}
\end{equation}
where $x(t)\in \R^n,$ $\dot x(t)={{\rm d}x\over {\rm d}t}(t)$, $n\ge 2$. 

When $n=3$ then  equation \eqref{1.1} is the equation of motion of a nonrelativistic particle of mass $m=1$ and charge $e=1$ in an
external and static electric (or gravitational) field described by $V$ (see \cite{L}) where $x$ denotes the position of the particle, 
$\dot x$ denotes its velocity and $\ddot x$ denotes its acceleration and $t$ denotes the time.

We also assume throughout this paper that $V$ satisfies the following
conditions
\begin{equation}
F=F^l+F^s,\label{1.3}
\end{equation}
where $F^l:=-\nabla V^l$, $F^s:=-\nabla V^s$ and $(V^l,V^s)\in (C^2(\R^n,\R))^2$, and where $V^l$ satisfies the following long range assumptions 
\begin{equation}
|\pa^j_x V^l(x)| \le \beta_{|j|}^l(1+|x|)^{-(\alpha+|j|)},\label{1.4a}
\end{equation}
and $V^s$ satisfies the following short range assumptions 
\begin{equation}
|\pa^j_x V^s(x)| \le \beta_{|j|+1}^s(1+|x|)^{-(\alpha+1+|j|)},\label{1.4b}
\end{equation}
for $x\in \R^n$ and $|j| \le 2$ and for some $\alpha\in (0,1]$ (here $j$ is the multiindex $j\in (\N \cup \{0\})^n, |j|= \sum_{m=1}^n j_m$, and $\beta_m^l$ and $\beta_{m'}^s$ are
positive real constants for $m=0,1,2$ and for $m'=1,2,3$). 
Note that the assumption $0<\alpha\le 1$ includes the decay rate of a Coulombian potential at infinity. Indeed for a Coulombian potential $V(x)={1\over |x|}$, estimates \eqref{1.4a} are satisfied uniformly for $|x|>\ep$ and $\alpha=1$ for any $\ep>0$.  Although our potentials $(V^l,V^s)$ are assumed to be $C^2$ on the entire space $\R^n$, our present work may provide interesting results even in presence of singularities for the potentials $(V^l,V^s)$.

For equation \eqref{1.1} the energy 
\begin{equation}
E={1\over 2}|\dot x(t)|^2+V(x(t))\label{1.2}
\end{equation}
is an integral of motion.

For $\sigma>0$ we denote by $\B(0,\sigma)$ the Euclidean open ball of center 0 and radius $\sigma$, $\B(0,\sigma)=\{y\in \R^n\ |\ |y|< \sigma\}$, and we denote by $\overline{\B(0,\sigma)}=\{y\in \R^n\ |\ |y|\le \sigma\}$ its closure.
We set $\mu:=\sqrt{2^5n\max(\beta_1^l,\beta_2^l)\over \alpha}$.
Under conditions \eqref{1.4a} the following is valid (see Lemma \ref{lem_scatinit} given in the next Section):
for any 
$v\in \R^n\b\B(0,\mu),$ there exists a unique solution $z_\pm(v,.)$ of the equation
\begin{equation}
\ddot z(t)=F^l(z(t)),\ t\in \R,\label{5.102}
\end{equation}
so that
\begin{equation*}
\dot z_\pm(v,t)-v=o(1),\textrm{ as }t\to \pm\infty,\  z_\pm(v,0)=0,
\end{equation*}
and
\begin{equation*}
|z_\pm(v,t)-tv|\le {2^{5\over 2}n^{1\over 2}\beta_1^l\over \alpha|v|}|t|\textrm{ for }t\in \R.
\end{equation*}
When $F^l\equiv0$ then $\beta_1^l$, $\beta_2^l$ and $\mu$ can be arbitrary close to 0, and we have that $z_\pm(v,t)=tv$ for $(t,v)\in \R\times \R^n$, $v\not=0$.

Then 
under conditions \eqref{1.4a} and \eqref{1.4b}, the following is valid: for any 
$(v_-,x_-)\in \R^n\b\B(0,\mu)\times\R^n,$ 
the equation \eqref{1.1}  has a unique solution $x\in C^2(\R,\R^n)$ such that
\begin{equation}
{x(t)=z_-(v_-,t)+x_-+y_-(t),}\label{1.6}
\end{equation}
where $|\dot y_-(t)|+|y_-(t)|\to 0,$ as $t\to -\infty;$  in addition for almost any 
$(v_-,x_-)\in \R^n\b\B(0,\mu)\times \R^n$,
\begin{equation}
{x(t)=z_+(v_+,t)+x_++y_+(t),}\label{1.7}
\end{equation}
for a unique $(v_+,x_+)\in \R^n\times\R^n$, where $|v_+|=|v_-|\ge \mu$ by conservation of the energy \eqref{1.2}, and where $v_+=:a(v_-,x_-)$, $x_+=:b(v_-,x_-),$ and  
$|\dot y_+(t)|+|y_+(t)|\to 0,$ as $t \to +\infty$. A solution $x$ of \eqref{1.1} that satisfies \eqref{1.6} and \eqref{1.7} for some $(v_-,x_-)$, $v_-\not=0$, is called a scattering solution. 

We call the map 
$S: (\R^n\b \B(0,\mu))\times\R^n \to (\R^n\b\B(0,\mu))\times\R^n$
given by the formulas
\begin{equation}
{v_+=a(v_-,x_-),\ x_+=b(v_-,x_-)},\label{1.8}
\end{equation}
the scattering map for the equation \eqref{1.1}. In addition, $a(v_-,x_-),$ $b(v_-,x_-)$ are called the scattering data for
the equation \eqref{1.1}, and we define
\begin{equation}
a_{sc}(v_-,x_-)=a(v_-,x_-)-v_-,\ \ b_{sc}(v_-,x_-)=b(v_-,x_-)-x_-.\label{1.10}
\end{equation}
Our definition of the scattering map is derived from constructions given in \cite{He, DG}. We refer the reader to \cite{He, DG} and references therein for the forward classical scattering theory.

By ${\cal D}(S)$ we denote the set of definition of $S$. Under the conditions \eqref{1.4a} and \eqref{1.4b} the map  $S:{\cal D}(S)\to (\R^n\b \B(0,\mu))\times\R^n$ is continuous, 
and ${\rm Mes}(((\R^n\b \B(0,\mu))\times \R^n) \b {\cal D}(S))=0$ for the Lebesgue
measure on $\R^n \times \R^n$. 
In addition the map $S$ is uniquely determined by its restriction to ${\cal M}(S)={\cal
D}(S)\cap {\cal M}$ and by $F^l$, where
${\cal M}=\{(v_-,x_-)\in \R^n \times \R^n\ |\ v_-\neq 0, v_-\cdot x_-=0\}$.
(Indeed if $x(t)$ is a solution of \eqref{1.1} then $x(t+t_0)$ is also a solution of \eqref{1.1} for any $t_0\in \R$.)

One can imagine the following experimental setting that allows to measure the scattering data without knowing the potential $V$ inside a (a priori bounded) region of interest.
First find a potential $V^l$ that generates the same long range effects as $V$ does. Then compute the solutions $z_\pm(v,.)$ of equation \eqref{5.102}. Then for a fixed $(v_-,x_-)\in (\R^n\b \B(0,\mu))\times \R^n$ send a particle far away from the region of interest  with a trajectory asymptotic to $x_-+z_-(v_-,.)$ at large and negative times.  When the particle escapes any bounded region of the space at finite time, then detect the particle  and find $S(v_-,x_-)=(v_+,x_+)$ so that  the trajectory of the particle is  asymptotic to  $x_++z_+(v_+,.)$ at large and positive times far away from the bounded region of interest.

In this paper we consider the following inverse scattering problem for equation \eqref{1.1}:
\begin{equation}
\textrm{Given }S\textrm{ and given the long range tail }F^l\textrm{ of the force }F,\textrm{ find }F^s. \label{P}
\end{equation}

The main results of the present work consist in estimates and asymptotics for the scattering data $(a_{sc},b_{sc})$ 
and scattering solutions for the equation \eqref{1.1} and in application of these asymptotics and estimates to the inverse scattering problem \eqref{P} at high energies. Our main results include, in particular, Theorem \ref{thm} given below that provides the high energies asymptotics of the scattering data and the Born approximation of the scattering data at fixed energy.

Consider
\begin{equation*}
\label{1.9}
T\S^{n-1}:=\{(\theta,x)\in \S^{n-1}\times\R^n\ |\ \theta\cdot x=0\},
\end{equation*}
and for any $m\in \N$ consider the x-ray transform $P$ defined by
\begin{equation*}
Pf(\theta,x):=\int_{-\infty}^{+\infty}f(t\theta+x)dt
\end{equation*}
for any function $f\in C(\R^n,\R^m)$ so that $|f(x)|=O(|x|^{-\tilde \beta})$ as $|x|\to +\infty$ for some $\tilde \beta>1$. 
For $(\sigma,\tilde \beta,r,\tilde \alpha)\in (0,+\infty)^2\times(0,1)\times (0,1]$, let $s_0=s_0(\sigma,r,\tilde \beta,\tilde \alpha)$ be defined as the root of the equation
\begin{equation}
1={4\tilde \beta n(\sigma+1)\over \tilde \alpha r({s_0\over 2^{3\over 2}}-r)(1-r)^{\tilde \alpha+2}}\big(1+{1\over {s_0\over 2^{3\over 2}}-r}\big)^2,\ s_0>2^{3\over 2}r.\label{1.30}
\end{equation}
Then we have the following results.

\begin{theorem}
\label{thm}
Let $(\theta,x)\in T\S^{n-1}$. Under conditions \eqref{1.4a} and \eqref{1.4b} the following limits are valid
\begin{eqnarray}
\lim_{s\to +\infty}s a_{sc}(s\theta,x)&=&PF^l(\theta,x)+PF^s(\theta,x),\label{t1}\\
\lim_{s\to +\infty}s^2\theta\cdot\Big(b_{sc}(s\theta,x)-W(s\theta,x)\Big)&=&-PV^s(\theta,x),\label{t2}
\end{eqnarray}
where
\begin{eqnarray}
W(v,x)&:=&\int_{-\infty}^0\int_{-\infty}^\sigma \big(F^l(z_-(v,\tau)+x)-F^l(z_-(v,\tau))\big)d\tau d\sigma\label{t0}\\
&&-\int_0^{+\infty}\int_\sigma^{+\infty}\big(F^l(z_+(a(v,x),\tau)+x)-F^l(z_+(a(v,x),\tau))\big)d\sigma d\tau\Big),\nonumber
\end{eqnarray}
for $(v,x)\in {\cal D}(S)$.

In addition, 
\begin{equation}
\big|a_{sc}(s\theta,x)-\int_{-\infty}^{+\infty}F(\tau s\theta+x)d\tau\big|\le
{4n^2(3|x|+5)\beta^2\big(1+{1\over 2^{-{3\over 2}}s-r}\big)^2\over  \alpha^2(1-r)^{2\alpha+3}(2^{-{3\over 2}}s-r)^2},\label{t3}
\end{equation}
\begin{eqnarray}
\big|b_{sc}(s\theta,x)-W(s\theta,x)-\int_{-\infty}^0\int_{-\infty}^\sigma F^s(\tau s\theta+x)d \tau d\sigma\nonumber\\
+\int_0^{+\infty}\int_\sigma^{+\infty}F^s(\tau s\theta+x)d \tau d\sigma\big|\le {4n^2(3|x|+5)\beta^2\big(1+{1\over 2^{-{3\over 2}}s-r}\big)^2\over \alpha^2(1-r)^{2\alpha+2}(2^{-{3\over 2}}s-r)^3},\label{t4}
\end{eqnarray}
for $(r,(\theta,x))\in (0,1)\times T\S^{n-1}$ and for $s>s_0(|x|,r,\beta,\alpha)$, 

\noindent where $\beta=\max(\beta_1^l,\beta_2^l,\beta_2^s,\beta_3^s)$.
\end{theorem}

Note that the vector $W$ defined by \eqref{t0} is known from the scattering data and from $F^l$.
Then from \eqref{t1} (resp. \eqref{t2}) and inversion formulas for the X-ray transform for $n\ge 2$ (see \cite{R, GGG, Na, No}) it follows that $F^s$ can be  reconstructed from  $a_{sc}$ (resp. $b_{sc}$). 

Note that \eqref{t3} and \eqref{t4} also
give the asymptotics of  $a_{sc},$ $b_{sc}$, when the parameters $\alpha,$ $n,$ $s$, $\theta$ and $x$ are fixed and
$\beta$ decreases to $0$ (where 
$\beta=\max(\beta_1^l,\beta_2^l,\beta_2^s,\beta_3^s)$). In that regime the leading term of $sa_{sc}(s\theta,x)$ and $s^2\theta\cdot(b_{sc}(s\theta,x)-W(s\theta,x))$ for $(\theta,x)\in T\S^{n-1}$ and for $s> s_0(|x|,r,\beta,\alpha)$ is given by the right hand sides of \eqref{t1} and \eqref{t2} respectively.
Therefore Theorem \ref{thm} gives the Born approximation for the scattering data at fixed energy when
the potential is sufficiently weak, and it proves that $F^s$ can be reconstructed from the Born approximation of the scattering map at fixed energy.

Theorem \ref{thm} is a generalization of \cite[formulas (4.8a), (4.8b), (4.9a) and (4.9b)]{No} where inverse scattering for the classical multidimensional Newton equation was studied in the short range case ($F^l\equiv 0$). We develop Novikov's framework \cite{No} to obtain our results.
Note that results \cite[formulas  (4.8b) and (4.9b)]{No}  also provide the approximation of the scattering data
$(a_{sc}(v_-,x_-),b_{sc}(v_-,x_-))$ for the short range case ($F^l\equiv 0$) when the parameters $\alpha,$ $n,$ $v_-$ and $\beta$ are fixed and
$|x_-|\to +\infty$. Such an asymptotic regime is not covered by Theorem \ref{thm}. Therefore we shall modify in Section 3 the definition of the scattering map 
to study these modified scattering data in the following three asymptotic regimes: at high energies, Born approximation at fixed energy, and when the parameters $\alpha,$ $n$, $v_-$ and $\beta$ are fixed and
$|x_-|\to +\infty$.

For inverse scattering at fixed energy for the multidimensional Newton equation, see for example \cite{Jo} and references therein.

For the inverse scattering problem in quantum mechanics for the Schr\"odinger equation, see  for example 
\cite{F}, \cite{EW}, \cite{No2} and references given in \cite{No2}.

Our paper is organized as follows. In Section 2 we transform the differential equation \eqref{1.1} with initial conditions \eqref{1.6} in
an integral equation which takes the form $y_-=A(y_-).$ Then we study the operator $A$ on a
suitable space  (Lemma \ref{lem_cont}) and we give estimates for the deflection $y_-(t)$ in \eqref{1.6} and for  
the scattering data $a_{sc}(v_-,x_-),b_{sc}(v_-,x_-)$  (Theorem \ref{thm_y}). We prove Theorem \ref{thm}.
Note that we work with small angle scattering compared to the dynamics generated by $F^l$ through the ``free" solutions $z_-(v_-,t)$: In particular, the angle between the vectors $
\dot x(t)=\dot z_-(v_-,t)+\dot y_-(t)$ and $\dot z_-(v_-,t)$ goes to zero
when the parameters $\beta$, $\alpha$, $n$, $v_-/|v_-|$, $x_-$ are fixed and $|v_-|$ increases.
We also provide similar results when one replaces the ``free" solutions $z_-(v,.)$ by some other functions ``$z_{-,N}(v,.)$" that may be easier to compute in practise (Formulas \eqref{t1N} and \eqref{t2N}).
In Section 3 we change the definition of the scattering map so that one can obtain for the modified scattering data $(\tilde a_{sc}(v_-,x_-),\tilde b_{sc}(v_-,x_-))$ their approximation  at high energies, or their Born approximation at fixed energy, or their approximation  when the parameters $\alpha,$ $n,$ $v_-$ and $\beta$ are fixed and
$|x_-|\to +\infty$ (Theorems \ref{thm_y2}, \ref{thm2} and formulas \eqref{t6N} and \eqref{t7N}). 
Sections 4, 5, 6 and 7 are devoted to proofs of our Theorems and Lemmas. 

\section{Scattering solutions}
\label{cont}
\subsection{Integral equation}
First we need the following Lemma \ref{lem_scatinit} that generalizes the statements given in the Introduction on the existence of peculiar solutions $z_\pm$ of the equation \eqref{5.102}.

\begin{lemma}
\label{lem_scatinit}
Assume conditions \eqref{1.4a}. Let $(v,x,w,h)\in (\R^n)^4$ so that $v\cdot x=0$ and
\begin{equation}
|w-v|\le {|v|\over 4\sqrt{2}}\textrm{ and } |h|< 1+{|x|\over \sqrt{2}},\label{5.100}
\end{equation}
and assume
\begin{equation}
{2^5n\max(\beta_1^l,\beta_2^l)\over \alpha|v|^2(1+{|x|\over \sqrt{2}}-|h|)^\alpha}\le 1.\label{5.101}
\end{equation}
Then there exists a unique solution $z_\pm(w,x+h,.)$ of the equation \eqref{5.102}
so that
\begin{equation}
\dot z_\pm(w,x+h,t)-w=o(1),\textrm{ as }t\to \pm\infty,\  z_\pm(w,x+h,0)=x+h,\label{5.103a}
\end{equation}
and
\begin{equation}
|z_\pm(w,x+h,t)-x-h-tw|\le {2^{5\over 2}n^{1\over 2}\beta_1^l\over \alpha|v|(1+{|x|\over \sqrt{2}}-|h|)^\alpha}|t|,\label{5.103b}
\end{equation}
for $t\in \R$.
\end{lemma}

A proof of Lemma \ref{lem_scatinit} is given in Section \ref{proofprel}.

For the rest of this Section we set
\begin{eqnarray}
\mu&:=&\sqrt{2^5n\max(\beta_1^l,\beta_2^l)\over \alpha},\label{0.7a}\\
z_\pm(v,t)&=&z_\pm(v,0,t)\textrm{ for }t\in \R, \textrm{ when }  \mu\le |v|,\label{0.7b}\\
\beta_2&:=&\max(\beta_2^l,\beta_2^s).\label{0.7c}
\end{eqnarray}

Let $(v_-,x_-)\in \R^n\times \R^n$, $v_-\cdot x_-=0$ and $|v_-|\ge \mu$.
Then the function $y_-$ in \eqref{1.6} satisfies the integral equation
$y_-=A(y_-)$
where
\begin{equation}
A(f)(t)=\int_{-\infty}^t\int_{-\infty}^\sigma\Big(F(z_-(v_-,\tau)+x_-+f(\tau))-F^l(z_-(v_-,\tau))\Big)d\tau d\sigma\label{2.1}
\end{equation}
for $t\in \R$ and for $f\in C(\R,\R^n)$, $\sup_{(-\infty,0]}|f|<\infty$.
Under conditions \eqref{1.4a} and \eqref{1.4b} we have $A(f)\in C^2(\R,\R^n)$ for $f\in C(\R,\R^n)$ so that $\sup_{(-\infty,0]}|f|<\infty$.

For $r>0$ we introduce the following complete metric space $M_r$ defined by
\begin{equation}
M_r=
\lbrace
f\in C(\R,\R^n)\ |\ \sup_{(-\infty,0]}|f|+\sup_{t\in[0,+\infty)}\Big({|f(t)|\over 1+|t|}\Big)\le r
\rbrace,\label{2.3}
\end{equation}
and endowed with the norm $\|.\|$ where $\|f\|=\sup_{(-\infty,0]}|f|+\sup_{t\in[0,+\infty)}\Big({|f(t)|\over 1+|t|}\Big)$.
Then we have the following estimate and contraction estimate for the map $A$ restricted to $M_r$.

\begin{lemma}
\label{lem_cont}
Let $(v_-,x_-)\in (\R^n\b \B(0,\mu))\times\R^n$, $v_-\cdot x_-=0$, and let $r>0$, $r<\max({|v_-|\over 2^{3\over 2}},1+{|x_-|\over \sqrt{2}})$.
Then the following estimates are valid:
\begin{eqnarray}
\|A(f)\|&\le& \rho(n,\alpha,\beta_2,|x_-|,|v_-|,r)\label{l1}\\
&:=&{\beta_2(n(3|x_-|+2r)+2\sqrt{n})\over ({|v_-|\over 2^{3\over 2}}-r)(1-r)^\alpha}\big({2\over \alpha({|v_-|\over 2^{3\over 2}}-r)}
+{1\over (\alpha+1)(1-r)}\big),\nonumber
\end{eqnarray}
and
\begin{equation}
\|A(f_1)-A(f_2)\|\le \lambda(n,\alpha,\beta_2,\beta_3^s,|x_-|,|v_-|,r)\|f_1-f_2\|\label{l2},
\end{equation}
\begin{eqnarray}
 \lambda(n,\alpha,\beta_2,\beta_3^s,|x_-|,|v_-|,r)&:=&{2n\over \alpha({|v_-|\over 2^{3\over 2}}-r)(1-r+{|x_-|\over \sqrt{2}})^\alpha}\big(\beta_2+{\beta_3^s\over 1-r+{|x_-|\over \sqrt{2}}}+{\beta_3^s\over {|v_-|\over 2^{3\over2}}-r}\big)\nonumber\\
&&\times \big({1\over 1-r+{|x_-|\over \sqrt{2}}}+{1\over {|v_-|\over 2^{3\over 2}}-r}\big),\nonumber
\end{eqnarray}
for $(f,f_1,f_2)\in M_r^3$.
\end{lemma}
A proof of Lemma \ref{lem_cont} is given in Section \ref{proof_cont}.

We also need the following result.
\begin{lemma}
\label{lem_decomp}
Let $(v_-,x_-)\in (\R^n\b \B(0,\mu))\times\R^n$, $v_-\cdot x_-=0$, and let $r>0$, $r<\max({|v_-|\over 2^{3\over 2}},1+{|x_-|\over \sqrt{2}})$. When $y_-\in M_r$ is a fixed point for the map $A$ then $z_-(v_-,.)+x_-+y_-(.)$ is a scattering solution for equation \eqref{1.1} and
\begin{equation}
z_-(v_-,t)+x_-+y_-(t)
=z_+(a(v_-,x_-),t)+b(v_-,x_-)+y_+(t),\label{300}
\end{equation}
for $t\ge 0$, where
\begin{eqnarray}
a(v_-,x_-)&:=&v_-+\int_{-\infty}^{+\infty}F(z_-(v_-,\tau)+x_-+y_-(\tau))d\tau,\label{300a}\\
b(v_-,x_-)&:=&x_-+l(v_-,x_-,y_-)+l_1(v_-,x_-)+l_2(v_-,x_-,y_-),\label{300b}
\end{eqnarray}
\begin{equation}
y_+(t):=\int_t^{+\infty}\int_\sigma^{+\infty}\big(F(z_-(v_-,\tau)+x_-+y_-(\tau))-F^l(z_+(a(v_-,x_-),\tau))\big)d\tau d\sigma,\label{300c}
\end{equation}
for $t\ge 0$, and where
\begin{eqnarray}
l(v_-,x_-,y_-)&:=&\int_{-\infty}^0\int_{-\infty}^\sigma \left(F\big(z_-(v_-,\tau)+x_-+y_-(\tau)\big)-F^l\big(z_-(v_-,\tau)\big)\right)d\tau d\sigma\nonumber\\
&&-\int_0^{+\infty}\int_\sigma^{+\infty}F^s\big(z_-(v_-,\tau)+x_-+y_-(\tau)\big)d\tau d\sigma,\label{502a}
\end{eqnarray}
\begin{equation}
l_1(v_-,x_-):=-\int_0^{+\infty}\int_\sigma^{+\infty}\big(F^l(z_+(a(v_-,x_-),\tau)+x_-)-F^l(z_+(a(v_-,x_-),\tau))\big)d\tau d\sigma,\label{502b}
\end{equation}
\begin{equation}
l_2(v_-,x_-,y_-):=-\int_0^{+\infty}\int_\sigma^{+\infty}\big(F^l(z_-(v_-,\tau)+x_-+y_-(\tau))-F^l(z_+(a(v_-,x_-),\tau)+x_-)\big)d\tau d\sigma,\label{502c}
\end{equation}
for $t\ge 0$.
\end{lemma}
Lemma \ref{lem_decomp} is proved in Section \ref{proofprel}.
Note that $l_1$ is known from the scattering data and the knowledge of $F^l$.

\subsection{Estimates on the scattering solutions}
In this Section our main results consist in estimates and asymptotics for the scattering data $(a_{sc},b_{sc})$ 
and scattering solutions for the equation \eqref{1.1}.

\begin{theorem}
\label{thm_y}
Under the assumptions of Lemma \ref{lem_decomp} the following estimates are valid
\begin{eqnarray}
|\dot y_-(t)|&\le&{\beta_2\big(n(|x_-|+r)+\sqrt{n}\big)\over (\alpha+1)\big({|v_-|\over 2\sqrt{2}}-r\big)\Big(1-r+|t|\big({|v_-|\over 2\sqrt{2}}-r\big)\Big)^{\alpha+1}},\label{t10a}\\
|y_-(t)|&\le&{\beta_2(n(|x_-|+r)+\sqrt{n})\over \alpha(\alpha+1)\big({|v_-|\over 2\sqrt{2}}-r\big)^2\Big(1-r+|t|\big({|v_-|\over 2\sqrt{2}}-r\big)\Big)^\alpha}, \label{t10b}
\end{eqnarray}
for $t\le 0$.
In addition 
\begin{equation}
|a_{sc}(v_-,x_-)|\le {2n^{1\over 2}\over\big({|v_-|\over 2 \sqrt{2}}-r\big)(1+{|x_-|\over \sqrt{2}}-r)^\alpha}\Big({\beta_1^l\over \alpha}+{\beta_2\over (\alpha+1)(1+{|x_-|\over \sqrt{2}}-r)}\Big).\label{503a}
\end{equation}
\begin{equation}
|l(v_-,x_-,y_-)|\le{\beta_2n^{1\over 2} (n^{1\over 2}(|x_-|+r)+2)\over \alpha (\alpha+1)(1-r)^\alpha\big({|v_-|\over 2 \sqrt{2}}-r\big)^2},\label{503b}\\
\end{equation}
and
\begin{eqnarray}
&&\big|a_{sc}(v_-,x_-)-\int_{-\infty}^{+\infty}F(z_-(v_-,\tau)+x_-)d\tau\big|\nonumber\\
&\le&{4\max(\beta_2,\beta_3^s)^2n^{3\over 2}(n^{1\over 2}(3|x_-|+2r)+2)\over \alpha^2({|v_-|\over 2^{3\over 2}}-r)^2(1-r)^{2\alpha+3}}\big({1\over{|v_-|\over 2^{3\over 2}}-r}+1\big)^2,\label{504a}
\end{eqnarray}
\begin{eqnarray}
|l(v_-,x_-,y_-)-l(v_-,x_-,0)|&\le&{4\max(\beta_2,\beta_3^s)^2n^{3\over 2}(n^{1\over 2}(3|x_-|+2r)+2)\over \alpha^2(\alpha+1)({|v_-|\over 2^{3\over 2}}-r)^3(1-r)^{2\alpha+2}}\nonumber\\
&&\times\big({1\over {|v_-|\over 2^{3\over 2}}-r}+1\big)^2.\label{504b}
\end{eqnarray}
In addition when
\begin{equation}
{8n\max(\beta_1^l,\beta_2)\over \alpha \big({|v_-|\over 2^{3\over 2}}-r\big)^2(1-r)^{\alpha+1}}\le 1,\label{503}
\end{equation}
then
\begin{eqnarray}
|l_1(v_-,x_-)|&\le&{8\beta_2 n|x_-|\over \alpha(\alpha+1)|v_-|^2},\label{503c}\\
|l_2(v_-,x_-,y_-)|&\le&{2n^{3\over 2}\beta_2^2(n^{1\over 2}(2|x_-|+r)+3)\over \alpha^2(\alpha+1)^2(1-r)^{2\alpha}({|v_-|\over 2^{3\over 2}}-r)^4},\label{503d}\\
|y_+(t)|&\le& {2n^{1\over 2}\beta_2\over \alpha(\alpha+1)({|v_-|\over 2^{3\over 2}}-r)^2(1-r+{|x_-|\over \sqrt{2}}+t({|v_-|\over 2^{3\over 2}}-r))^\alpha}\nonumber\\
&&\times\big(1+{2n\beta_2(n^{1\over 2}(2|x_-|+r)+3)\over \alpha(\alpha+1)({|v_-|\over 2^{3\over 2}}-r)^2(1-r)^\alpha}\big),\label{503e}
\end{eqnarray}
for $t\ge 0$.
\end{theorem}
A proof of Theorem \ref{thm_y} is given in Section \ref{proof_thm_y}.
We now prove Theorem \ref{thm} combining Theorem \ref{thm_y}, Lemma \ref{lem_cont} and estimate \eqref{5.103b}.

\subsection{Proof of Theorem \ref{thm}}
Let  $(v_-,x_-)\in \R^n\times\R^n$, $v_-\cdot x_-=0$ and $|v_-|\ge \mu$. 
We first prove estimates \eqref{901} and \eqref{902} given below. We use the following estimate \eqref{505a}
\begin{eqnarray}
&&|x_-+\eta v_-\tau+(1-\eta)z_-(v_-,\tau)|\ge |x_-+\tau v_-|-|z_-(v_-,\tau)-\tau v_-|\nonumber\\
&\ge&{|x_-|\over \sqrt{2}}+({|v_-|\over \sqrt{2}}-{2^{5\over 2}n^{1\over 2}\beta_1^l\over \alpha|v_-|})|\tau|
\ge {|x_-|\over \sqrt{2}}+|\tau| {|v_-|\over 2^{3\over 2}},\label{505a}
\end{eqnarray}
for $\eta\in (0,1)$ and $\tau \in \R$ (we used \eqref{5.103b} and \eqref{5.101}).
Then from \eqref{1.4a}, \eqref{1.4b}, \eqref{505a} and \eqref{5.103b} it follows that
\begin{eqnarray}
|F^s(z_-(v_-,\tau)+x_-)-F^s(\tau v_-+x_-)|&\le&\sup_{\eta\in (0,1)}{n\beta_3^s|z_-(v_-,\tau)-\tau v_-|\over(1+|x_-+\eta v_-+(1-\eta)z_-(v_-,\tau)|)^{\alpha+3}}\nonumber\\
&\le&{2^{5\over 2}n^{3\over 2}\beta_3^s\beta_1^l|\tau|\over \alpha|v_-|(1+{|x_-|\over \sqrt{2}}+|\tau| {|v_-|\over 2^{3\over 2}})^{\alpha+3}},\label{900a}
\end{eqnarray}
for $\tau\in \R$. Similarly
\begin{equation}
|F^l(z_-(v_-,\tau)+x_-)-F^l(\tau v_-+x_-)|\le{2^{5\over 2}n^{3\over 2}|\tau|\beta_2\beta_1^l\over \alpha|v_-|(1+{|x_-|\over \sqrt{2}}+|\tau| {|v_-|\over 2^{3\over 2}})^{\alpha+2}},\label{900b}
\end{equation}
for $\tau\in \R$. Then using \eqref{900a} and \eqref{900b} we have 
$$
\big|\int_{-\infty}^{+\infty}F(z_-(v_-,\tau)+x_-)d\tau-\int_{-\infty}^{+\infty}F(\tau v_-+x_-)d\tau\big|
$$
\begin{equation}
\le {2^{7\over 2}n^{3\over 2}\max(\beta_2,\beta_3^s)\beta_1^l\over \alpha|v_-|}\int_{-\infty}^{+\infty}{|\tau|d\tau\over (1+{|x_-|\over \sqrt{2}}+{|v_-|\over 2^{3\over 2}}|\tau|)^{\alpha+2}}\le {2^{15\over 2}n^{3\over 2}\max(\beta_2,\beta_3^s)\beta_1^l\over \alpha^2|v_-|^3(1+{|x_-|\over \sqrt{2}})^\alpha}.\label{901}
\end{equation}
Set \begin{equation}
\Delta_1(v_-,x_-)=\int_{-\infty}^0\int_{-\infty}^\sigma F^s\big(z_-(v_-,\tau)+x_-\big)d\tau d\sigma-\int_0^{+\infty}\int_\sigma^{+\infty}F^s\big(z_-(v_-,\tau)+x_-\big)d\tau d\sigma.
\label{3.8}
\end{equation}
Then using \eqref{900a} we have
$$
\big|\Delta_1(v_-,x_-)-\int_{-\infty}^0\int_{-\infty}^\sigma F^s\big(\tau v_-+x_-\big)d\tau d\sigma+\int_0^{+\infty}\int_\sigma^{+\infty}F^s\big(\tau v_-+x_-\big)d\tau d\sigma\big|
$$
\begin{equation}
\le {2^{7\over 2}n^{3\over 2}\beta_1^l\beta_3^s\over \alpha|v_-|}\int_{-\infty}^0\int_{-\infty}^\sigma{|\tau|d\tau d\sigma\over(1+{|x_-|\over \sqrt{2}}+{|v_-|\over 2\sqrt{2}}|\tau|)^{\alpha+3}}\le {2^{13\over 2}n^{3\over 2}\beta_1^l\beta_3^s\over \alpha^2(\alpha+1)|v_-|^3(1+{|x_-|\over \sqrt{2}})^{\alpha}}.\label{902}
\end{equation}
Let $r>0$, $r<\max({|v_-|\over 2^{3\over 2}},1)$. Note that
\begin{equation}
\max\big({\rho\over r}, \lambda,{8n\max(\beta_1^l,\beta_2)\over \alpha \big({|v_-|\over 2^{3\over 2}}-r\big)^2(1-r)^{\alpha+1}} \big)\le
{4\beta n(|x_-|+1)\over \alpha r({|v_-|\over 2^{3\over 2}}-r)(1-r)^{\alpha+2}}\big(1+{1\over {|v_-|\over 2^{3\over 2}}-r}\big)^2,\label{903}
\end{equation}
where $\rho$ and $\lambda$ are defined by \eqref{l1} and \eqref{l2} respectively.
Assume that $|v_-|> s_0(|x_-|,r,\beta,\alpha)$ where $s_0$ is the root of the equation \eqref{1.30}.
Then from \eqref{1.30} and Lemma \ref{lem_cont} it follows that $A$ has a unique fixed point in $M_r$ denoted by $y_-$. Then adding \eqref{504a} and \eqref{901} we obtain \eqref{t3}. Note also that 
\begin{eqnarray*}
&&l(v_-,x_-,0)=\int_{-\infty}^0\int_{-\infty}^\sigma \left(F^l\big(z_-(v_-,\tau)+x_-\big)-F^l\big(z_-(v_-,\tau)\big)\right)d\tau d\sigma\\
&&+\Delta_1(v_-,x_-).
\end{eqnarray*}
Hence adding \eqref{902}, \eqref{503d} and \eqref{504b} we obtain \eqref{t4}. Theorem \ref{thm} is proved.
\hfill$\Box$

\subsection{Motivations for changing the definition of the scattering map}
\label{com}
For a solution  $x$ at a nonzero energy for equation \eqref{1.1} we say that it is a scattering solution when there exists $\ep>0$ so that $1+|x(t)|\ge \ep(1+|t|)$ for $t\in \R$ (see \cite{DG}).
In the Introduction and in the previous subsections we choose to parametrize the scattering solutions of equation \eqref{1.1} by the solutions $z_\pm(v,.)$ of the equation \eqref{5.102} (see the asymptotic behaviors \eqref{1.6} and \eqref{1.7}), and then to formulate the inverse scattering problem \eqref{P} using this parametrization. 
To compute the "free" solutions $z_\pm(v,.)$ one has to integrate equation \eqref{5.102}. For some cases solving \eqref{5.102} leads to simple exact formulas (see \cite[Section 15]{L} when $F^l$ is a Coulombian force).  In general one may choose to approximate the solutions $z_\pm(v,.)$ by the functions $z_{\pm, N+1}(v,.)$ defined 
below. In general the functions $z_{\pm, N+1}(v,.)$ are easier to compute, and in this Subsection we use these approximations to obtain an other formulation of the inverse scattering problem and to mention results similar to Theorem \ref{thm} and to those given in the previous subsections.

Assume without loss of generality that $\alpha\not\in \{{1\over m}\ |\ m\in \N,\ m>0\}$ and set $N=\lfloor \alpha^{-1}\rfloor$ the integer part of $\alpha^{-1}$.
Then let $(x,v,w,h)\in(\R^n)^4 $  so that $v\cdot x=0$ and  \eqref{5.100} and \eqref{5.101} are satisfied. 
We define by induction (see also \cite{He})
\begin{eqnarray}
z_{\pm,0}(w,x+h,t)=x+h+tw,\label{pr1}\\
z_{-,m+1}(w,x+h,t)=x+h+tw+\int_0^t\int_{-\infty}^\sigma F^l(z_{-,m}(w,x+h,\tau))d\tau,\label{pr2}\\
z_{+,m+1}(w,x+h,t)=x+h+tw-\int_0^t\int_\sigma^{+\infty} F^l(z_{+,m}(w,x+h,\tau))d\tau,\label{pr3}
\end{eqnarray}
for $t\in \R$ and for $m=0\ldots N$.
Then one can prove the following estimates by induction (see the proof of Lemma \ref{lem_scatinit} and see also \cite{He})
\begin{equation}
|z_{\pm,m}(w,x+h,t)-wt-x-h|\le {2^{5\over 2}n^{1\over 2}\beta_1^l\over \alpha|v|(1+{|x|\over \sqrt{2}}-|h|)^\alpha}|t|,\ t\in \R,\ m=1\ldots N+1,\label{pr4}
\end{equation}
\begin{eqnarray}
&&|z_{\pm,m+1}(w,x+h,t)-z_{\pm,m}(w,x+h,t)|\le{2^{3(m+1)}n^{m+{1\over 2}}(\beta_2^l)^m\beta_1^l\over \alpha^{m+1}|v|^{2m+2}\Pi_{j=1}^{m+1}(1-j\alpha)j}\nonumber\\
&&\times \Big(
\big(1+{|x|\over \sqrt{2}}-|h|+|t|{|v|\over 2^{3\over2}}\big)^{1-(m+1)\alpha}-\big(1+{|x|\over \sqrt{2}}-|h|\big)^{1-(m+1)\alpha}\Big),\label{pr5}
\end{eqnarray}
for $m=0\ldots N-1$ and for $\pm t\ge 0$,
\begin{equation}
|z_{\pm,N+1}(w,x+h,t)-z_{\pm,N}(w,x+h,t)|\le{2^{3(N+1)}n^{N+{1\over 2}}(\beta_2^l)^N\beta_1^l\over \alpha^{N+1}|v|^{2N+2}\Pi_{j=1}^{N+1}j|1-j\alpha|
\big(1+{|x|\over \sqrt{2}}-|h|\big)^{(N+1)\alpha-1}},\label{pr6}
\end{equation}
for $\pm t\ge 0$, and 
\begin{equation}
|z_{\pm,m+1}(w,x+h,t)-z_{\pm,m}(w,x+h,t)|\le{2^{4m+{5\over 2}}n^{m+{1\over 2}}(\beta_2^l)^m\beta_1^l\over \alpha^{m+1}|v|^{2m+1}(1+{|x|\over \sqrt{2}}-|h|)^{(m+1)\alpha}}|t|,\label{pr7}
\end{equation}
for $t\in \R$ and for $m=1\ldots N$.
We set $z_{\pm,m}(v,.):=z_{\pm,m}(v,0,.)$ for $m=1\ldots N+1$ and $|v|\ge \mu$. For $(v_-,x_-)\in \R^n\times\R^n$, $|v_-|\ge \mu$, there exists a unique solution $x(t)$ of equation \eqref{1.1} so that 
\begin{equation}
x(t)=x_-+z_{-,N+1}(v_-,t)+y_-(t),\ t\in \R\textrm{ and } \lim_{t\to-\infty}(|y_-(t)|+|\dot y_-(t)|)=0.\label{pr8}
\end{equation}
In addition when the solution $x$ in \eqref{pr8} is a scattering solution then there exists a unique $(v_+,x_+)\in \R^n\times\R^n$, $|v_+|=|v_-|$ so that
\begin{equation}
x(t)=x_++z_{+,N+1}(v_+,t)+y_+(t),\ t\in \R\textrm{ and } \lim_{t\to+\infty}(|y_+(t)|+|\dot y_+(t)|)=0.\label{pr9}
\end{equation}
In that case we define the scattering data $(a_N(v_-,x_-), b_N(v_-,x_-)):=(v_+,x_+)$, and we consider the inverse scattering problem
\begin{equation}
\textrm{Given }(a_N, b_N)\textrm{ and given the long range tail }F^l\textrm{ of the force }F,\textrm{ find }F^s. \label{P3}
\end{equation}
The function $y_-$ in \eqref{pr8} satisfies the following integral equation
$y_-=A_N(y_-)$
where
\begin{equation}
A_N(f)(t)=\int_{-\infty}^t\int_{-\infty}^\sigma\Big(F(z_{-,N+1}(v_-,\tau)+x_-+f(\tau))-F^l(z_{-,N}(v_-,\tau))\Big)d\tau d\sigma\label{2.1N}
\end{equation}
for $t\in \R$ and for $f\in C(\R,\R^n)$, $\sup_{(-\infty,0]}|f|<\infty$. Then with appropriate changes in the proof of Lemma \ref{lem_cont} we can study the operator $A_N$ restricted to $M_r$ and we can obtain estimate and contraction estimate similar to \eqref{l1} and \eqref{l2}.  We also obtain the analog of Lemma \ref{lem_decomp} by appropriate change in its proof, and the decomposition \eqref{300} remains valid by replacing $a$, $b$ and $z_-(v_-,\tau)+x_-+y_-(\tau)$ by $a_N$, $b_N$ and  $z_{-,N+1}(v_-,\tau)+x_-+y_-(\tau)$ in \eqref{300}-- \eqref{502c}, and by replacing $z_+$ and $F^l\big(z_-(v_-,\tau)\big)$ by $z_{+,N+1}$ and $F^l\big(z_{-,N}(v_-,\tau)\big)$ in \eqref{300}
and \eqref{502a}, and by replacing $z_+$ by $z_{+,N}$ in \eqref{300c}, \eqref{502b} and \eqref{502c}. An analog of Theorem \ref{thm_y} can be proved for the scattering solutions and scattering data $(a_N, b_N)$. Set $a_{sc,N}(v_-,x_-):=a_N(v_-,x_-)-v_-$ and $b_{sc,N}(v_-,x_-):=b_N(v_-,x_-)-x_-$. Finally the following high energies limits are valid.
Let $(\theta,x)\in T\S^{n-1}$, then
\begin{eqnarray}
\lim_{s\to +\infty}s a_{sc,N}(s\theta,x)&=&PF^l(\theta,x)+PF^s(\theta,x),\label{t1N}\\
\lim_{s\to +\infty}s^2\theta\cdot\Big(b_{sc,N}(s\theta,x)-W_N(s\theta,x)\Big)&=&-PV^s(\theta,x),\label{t2N}
\end{eqnarray}
where 
\begin{eqnarray}
W_N(v,x)&:=&\int_{-\infty}^0\int_{-\infty}^\sigma \big(F^l(z_{-,N}(v,\tau)+x)-F^l(z_{-,N}(v,\tau))\big)d\tau d\sigma\label{t0N}\\
&&-\int_0^{+\infty}\int_\sigma^{+\infty}\big(F^l(z_{+,N}(a_N(v,x),\tau)+x)-F^l(z_{+,N}(a_N(v,x),\tau))\big)d\sigma d\tau,\nonumber
\end{eqnarray}
for $(v,x)\in \R^n\times\R^n$, $v\cdot x=0$ and $|v|>C$ for some constant $C$.
The vector $W_N$ defined by \eqref{t0N} is known from the scattering data and from $F^l$.
For the Problem \eqref{P3}, from \eqref{t1N} (resp. \eqref{t2N}) and inversion formulas for the X-ray transform for $n\ge 2$ (see \cite{R, GGG, Na, No}) it follows that $F^s$ can be  reconstructed from  $a_{sc,N}$ (resp. $b_{sc,N}$). 

The limits  \eqref{t1N} and \eqref{t2N} follow from estimates similar to \eqref{t3} and \eqref{t4} that also
give the Born approximation of   $a_{sc,N},$ $b_{sc,N}$ at fixed energy. However these similar estimates also do not provide the asymptotics of the scattering data
$(a_{sc,N},b_{sc,N})$ when the parameters $\alpha,$ $n,$ $v_-$ and $\beta$ are fixed and $|x_-|\to +\infty$.
Motivated by this disadvantage, in the next section we modify the definition of the scattering map given in the Introduction so that one can obtain a result on this asymptotic regime.

\section{A modified scattering map}
\label{cont2}
\subsection{Changing the parametrization of the scattering solutions}
We set
\begin{equation}
\mu(\sigma):=\sqrt{2^5n\max(\beta_1^l,\beta_2^l)\over \alpha (1+{\sigma\over \sqrt{2}})^\alpha},\textrm{ for }\sigma\ge 0.\label{4.7a}
\end{equation}
Under conditions \eqref{1.4a} and \eqref{1.4b}, the following is valid: for any 
$(v_-,x_-)\in \R^n\b\{0\}\times\R^n$ so that $|v_-|\ge \mu(|x_-|)$ and $v_-\cdot x_-=0$, then
the equation \eqref{1.1}  has a unique solution $x\in C^2(\R,\R^n)$ such that
\begin{equation}
{x(t)=z_-(v_-,x_-,t)+y_-(t),}\label{4.6}
\end{equation}
where $|\dot y_-(t)|+|y_-(t)|\to 0$, as $t\to -\infty$, and where $z_-(v_-,x_-,.)$ is defined in Lemma \ref{lem_scatinit} (for "$(w,x,v,h)=(v_-,x_-,v_-,0)$").

In addition the function $y_-$ in \eqref{4.6} satisfies the integral equation
$y_-=\A(y_-)$
where
\begin{equation}
\A(f)(t)=\int_{-\infty}^{t}\int_{-\infty}^\sigma\Big(F(z_-(v_-,x_-,\tau)+f(\tau))-F^l(z_-(v_-,x_-,\tau))\Big)d\tau d\sigma\label{4.1}
\end{equation}
for $t\in \R$ and for $f\in C(\R,\R^n)$, $\sup_{(-\infty,0]}|f|<\infty$.
Under conditions \eqref{1.4a} and \eqref{1.4b} we have $\A(f)\in C^2(\R,\R^n)$ for $f\in C(\R,\R^n)$ so that $\sup_{(-\infty,0]}|f|<\infty$. 
We study the map $\A$ defined by \eqref{4.1} on the metric space $M_r$ defined by \eqref{2.3}.
Set 
\begin{equation}
\tilde k(v_-,x_-,f)=v_-+\int_{-\infty}^{+\infty}F\big(z_-(v_-,x_-,\tau)+f(\tau)\big)d\tau,\label{5.2a}
\end{equation}
for $f\in M_r$.
For the rest of the section we set $\beta_2=\max(\beta_2^l,\beta_2^s)$.

The following Lemma \ref{lem_cont2} is the analog of Lemma \ref{lem_cont}.

\begin{lemma}
\label{lem_cont2}
Let $(v_-,x_-)\in \R^n\times\R^n$, $v_-\cdot x_-=0$, $|v_-|\ge \mu(|x_-|)$, and let $r>0$, $r<\max({|v_-|\over 2^{3\over 2}},1+{|x_-|\over \sqrt{2}})$.
Then the following estimates are valid:
\begin{eqnarray}
\|\A(f)\|&\le&\tilde \rho(n,\alpha,\beta_2,|x_-|,|v_-|,r)\label{lb1}\\
&=&{2\beta_2n^{1\over 2}(n^{1\over 2}r+1)\over({|v-|\over 2^{3\over 2}}-r)(1-r+{|x_-|\over \sqrt{2}})^\alpha}\Big({1\over (\alpha+1)(1-r+{|x_-|\over \sqrt{2}})}+{2\over \alpha( {|v_-|\over 2^{3\over 2}}-r)}\Big),\nonumber
\end{eqnarray}
\begin{equation}
\|\A(f_1)-\A(f_2)\|\le\lambda(n,\alpha,\beta_2,\beta_3^s,|x_-|,|v_-|,r)\|f_1-f_2\|,\label{lb2}
\end{equation}
and
\begin{equation}
|\tilde k(v_-,x_-,f)-v_-|\le {2n^{1\over 2}\over\big({|v_-|\over 2 \sqrt{2}}-r\big)(1+{|x_-|\over \sqrt{2}}-r)^\alpha}\Big({\beta_1^l\over \alpha}+{\beta_2\over (\alpha+1)(1+{|x_-|\over \sqrt{2}}-r)}\Big),\label{5.9}
\end{equation}
for $(f,f_1,f_2)\in M_r^3$, where $\lambda$ is defined in \eqref{l2}.
\end{lemma}
Proof of Lemma \ref{lem_cont2} is given in Section \ref{proof_cont}.

Let $(v_-,x_-)\in \R^n\times\R^n$, $v_-\cdot x_-=0$, and let $r\in(0,\max\big({1\over 2}+{|x_-|\over 2^{3\over 2}},2^{-{3\over 2}}|v_-|\big))$. Assume that
\begin{equation}
{20n\max(\beta_1^l,\beta_2)\over \alpha \big({|v_-|\over 2\sqrt{2}}-r\big)^2\big({1\over 2}+{|x_-|\over 2^{3\over 2}}-r\big)^\alpha}\le 1.\label{5.9a}
\end{equation}
Then using Lemma \ref{lem_scatinit} and \eqref{5.9} one can consider the free solution $z_+(\tilde k(v_-,x_-,$ $f),x_-,.)$ (for "$(w,v,x,h)=(\tilde k(v_-,x_-,f),v_-,x_-,0)$") for $f\in M_r$.
In addition,
with appropriate changes in the proof of Lemma \ref{lem_decomp} one can prove that when $y_-\in M_r$ is a fixed point of the operator $\A$ then $z_-(v_-,x_-,.)+y_-(.)$ is a scattering solution of \eqref{1.1} (in the sense given in Section \ref{com}), and one can prove that the following decomposition is valid
\begin{equation}
z_-(v_-,x_-,t)+y_-(t)
=z_+(\tilde a(v_-,x_-),x_-+h,t)+\big(\G_{v_-,x_-}(h)-h\big)+H(v_-,x_-,y_-,h)(t),\label{7.2}
\end{equation}
for $t\ge 0$, where
\begin{eqnarray}
\tilde a(v_-,x_-)&=&v_-+\int_{-\infty}^{+\infty}F\big(z_-(v_-,x_-,\tau)+y_-(\tau)\big)d\tau.\label{5.2ba}\\
\G_{v_-,x_-}(h)&=&\tilde l(v_-,x_-,y_-)-\int_0^{+\infty}\int_\sigma^{+\infty}\left(F^l\big(z_-(v_-,x_-,\tau)+y_-(\tau)\big)\right.\nonumber\\
&&-\left.F^l\big(z_+(\tilde a(v_-,x_-),x_-+h,\tau)\big)\right)d\tau d\sigma,\label{7.4}\\
\tilde l(v_-,x_-,y_-)&=&\int_{-\infty}^0\int_{-\infty}^\sigma \left(F\big(z_-(v_-,x_-,\tau)+y_-(\tau)\big)-F^l\big(z_-(v_-,x_-,\tau)\big)\right)d\tau d\sigma\nonumber\\
&&-\int_0^{+\infty}\int_\sigma^{+\infty}F^s\big(z_-(v_-,x_-,\tau)+y_-(\tau)\big)d\tau d\sigma\label{5.2b}
\end{eqnarray}
and
\begin{eqnarray}
H(v_-,x_-,y_-,h)(t)&=&\int_t^{+\infty}\int_\sigma^{+\infty}\big(F\big(z_-(v_-,x_-,\tau)+y_-(\tau)\big)\nonumber\\
&&-F^l\big(z_+(\tilde a(v_-,x_-),x_-+h,\tau)\big)\big)d\tau d\sigma,\label{7.3}
\end{eqnarray}
for $t\ge 0$ and for $|h|\le {1\over 2}+{|x_-|\over 2\sqrt{2}}$.
We need the following Lemma.

\begin{lemma}
\label{lem_cont3}
Let $(v_-,x_-)\in \R^n\times\R^n$, $v_-\cdot x_-=0$ and let $r>0$, $r<{1\over 2}+{|x_-|\over 2^{3\over 2}}$. Under conditions \eqref{1.4a}, \eqref{1.4b} and \eqref{5.9a}, when $y_-\in M_r$ is a fixed point of the operator $\A$ then
 the following estimates are valid:
\begin{eqnarray}
|\G_{v_-,x_-}(h)|&\le&{\beta_2(6(nr+\sqrt{n})+n(1+{|x_-|\over\sqrt{2}}))\over 2\alpha(\alpha+1)\big({|v_-|\over 2\sqrt{2}}-r\big)^2
\big({1\over 2}+{|x_-|\over 2\sqrt{2}}-r\big)^\alpha}\label{7.8a}\\
&\le&{1\over 4}+{|x_-|\over 10\sqrt{2}},\label{7.8b}
\end{eqnarray}
\begin{equation}
|\G_{v_-,x_-}(h)-\G_{v_-,x_-}(h')|\le {16n\beta_2^l|h-h'|\over \alpha (\alpha+1)|v_-|^2\big({1\over 2}+{|x_-|\over 2\sqrt{2}}\big)^\alpha}\le {|h-h'|\over 10},\label{7.9}
\end{equation}
for $(h,h')\in \R^n\times\R^n$, $|h'|\le |h|\le {1\over 4}+{|x_-|\over 2^{5\over 2}}$.
\end{lemma}
Lemma \ref{lem_cont3} is proved in Section \ref{proof_thm_y2}.

Under the assumptions of Lemma \ref{lem_cont3} the map $\G_{v_-,x_-}$ is a ${1\over 10}$-contraction map from $\overline \B(0,{1\over 4}+{|x_-|\over 2^{5\over 2}})$ to $\overline \B(0,{1\over 4}+{|x_-|\over 2^{5\over 2}})$. We denote by $\tilde b_{sc}(v_-,x_-)$ its unique fixed point in $\overline \B(0,{1\over 4}+{|x_-|\over 2^{5\over 2}})$, and we set $\tilde b(v_-,x_-):=x_-+\tilde b_{sc}(v_-,x_-)$ and $\tilde a_{sc}(v_-,x_-):= \tilde a(v_-,x_-)-v_-$. The decomposition \eqref{7.2} becomes
\begin{eqnarray}
&&z_-(v_-,x_-,t)+y_-(t)
=z_+(\tilde a(v_-,x_-),\tilde b(v_-,x_-),t)+y_+(t),\label{7.2a}\\
&&y_+(t)=H(v_-,x_-,y_-,\tilde b_{sc}(v_-,x_-))(t),\label{7.2b}
\end{eqnarray}
for $t\ge 0$. The map $(\tilde a_{sc}, \tilde b_{sc})$ are our modified scattering data. The inverse scattering problem for equation \eqref{1.1} can now be formulated as follows 
\begin{equation}
\textrm{Given }(\tilde a_{sc},\tilde b_{sc})\textrm{ and }F^l,\textrm{ find }F^s.\label{P2}
\end{equation}

\subsection{Estimates and asymptotics of the modified scattering data}
Let $(v_-,x_-)\in \R^n\times\R^n$, $v_-\cdot x_-=0$ and let $r>0$, $r<{1\over 2}+{|x_-|\over 2^{3\over 2}}$ so that condition \eqref{5.9a} is fullfilled. Then set
\begin{equation}
\tilde W(v_-,x_-):=\int_{-\infty}^0F^l(z_-(v_-,x_-,\tau))d\tau+\int_0^{+\infty}F^l(z_+(\tilde a(v_-,x_-),x_-,\tau))d\tau.\label{t5}
\end{equation}
Note that $\tilde W$ is known from the modified scattering data and from $F^l$.
We obtain the following analog of Theorem \ref{thm_y}.

\begin{theorem}
\label{thm_y2}
Under the assumptions of Lemma \ref{lem_cont3} and under conditions \eqref{1.4a}, \eqref{1.4b} and \eqref{5.9a}, the following estimates are valid:
\begin{eqnarray}
|\dot y_-(t)|&\le &{\beta_2\big(nr+\sqrt{n}\big)\over (\alpha+1)\big({|v_-|\over 2\sqrt{2}}-r\big)\Big(1+{|x_-|\over \sqrt{2}}-r+|t|\big({|v_-|\over 2\sqrt{2}}-r\big)\Big)^{\alpha+1}} ,\label{7.2c}\\
|y_-(t)|&\le &{\beta_2\big(nr+\sqrt{n}\big)\over \alpha(\alpha+1)\big({|v_-|\over 2\sqrt{2}}-r\big)^2\Big(1+{|x_-|\over \sqrt{2}}-r+|t|\big({|v_-|\over 2\sqrt{2}}-r\big)\Big)^{\alpha}},
\label{7.2d}
\end{eqnarray}
for $t\le0$;
and
\begin{eqnarray}
|\tilde a_{sc}(v_-,x_-)|&\le &{6\sqrt{n}\max(\beta_1^l,\beta_2)\over \alpha \big({|v_-|\over 2\sqrt{2}}-r\big)\big(1+{|x_-|\over \sqrt{2}}-r\big)^\alpha},\label{7.14b}\\
|\tilde b_{sc}(v_-,x_-)|&\le&{4\beta_2(nr+\sqrt{n})\over \alpha(\alpha+1)\big({|v_-|\over 2\sqrt{2}}-r\big)^2
\big({1\over 2}+{|x_-|\over 2\sqrt{2}}-r\big)^\alpha},\label{7.14}
\end{eqnarray}
\begin{equation}
|y_+(t)|\le {2\beta_2\sqrt{n}\over\alpha(\alpha+1)({|v_-|\over 2^{3\over 2}}-r)^2({1\over2}+{|x_-|\over 2\sqrt{2}}-r+t({|v_-|\over 2^{3\over 2}}-r))^{\alpha}},\label{7.16b}
\end{equation}
for $t\ge 0$.
In addition
\begin{eqnarray}
&&|\tilde a_{sc}(v_-,x_-)- \tilde W(v_-,x_-)-\int_{-\infty}^{+\infty}F^s(z_-(v_-,x_-,\tau))d\tau|\nonumber\\
&\le&{4\max(\beta_2,\beta_3^s)^2n(nr+\sqrt{n})\over \alpha(\alpha+1)({|v_-|\over 2^{3\over 2}}-r)^2(1-r+{|x_-|\over \sqrt{2}})^{2\alpha+1}}\big(3+{2\over {|v_-|\over 2^{3\over 2}}-r}\big)^2,\label{7.14c}
\end{eqnarray}
\begin{equation}
|\tilde b_{sc}(v_-,x_-)-\tilde l(v_-,x_-,0)|
\le{10n(nr+\sqrt{n})\max(\beta_2,\beta_3^s)^2\big( 3+{1\over{|v_-|\over 2\sqrt{2}}-r}\big)^2\over\alpha^2(\alpha+1) \big({|v_-|\over 2\sqrt{2}}-r\big)^3\big(1+{|x_-|\over \sqrt{2}}-r\big)^{2\alpha}}.\label{6.44}
\end{equation}

\end{theorem}
Theorem \ref{thm_y2} is proved in Section \ref{proof_thm_y2}.

For $(\sigma,\tilde \beta,r,\tilde \alpha)\in (0,+\infty)^3\times(0,1]$, $r<{1\over2}+{\sigma\over 2^{3\over 2}}$, and let $\tilde s_0=\tilde s_0(\sigma,r,\tilde \beta,\tilde \alpha)$ be defined as the root of the equation
\begin{equation}
1={12\tilde \beta n\over \tilde \alpha r({\tilde s_0\over 2^{3\over 2}}-r)({1\over 2}-r+{\sigma\over 2^{3\over 2}})^{\tilde \alpha}}\big(1+{1\over {\tilde s_0\over 2^{3\over 2}}-r}\big)^2,\ \tilde s_0>2^{3\over 2}r.\label{1.30b}
\end{equation}
Then the high energies asymptotics of the modified scattering data $(\tilde a_{sc},\tilde b_{sc})$ are given in the following Theorem \ref{thm2}.
\begin{theorem}
\label{thm2}
Let $(\theta,x)\in T\S^{n-1}$. Under conditions \eqref{1.4a} and \eqref{1.4b} the following limits are valid:
\begin{eqnarray}
&&\lim_{s\to +\infty}s(\tilde a_{sc}(s\theta,x)-\tilde W(s\theta,x))=PF^s(\theta,x),\label{t6}\\
&&\lim_{s\to +\infty}s^2\theta\cdot \tilde b_{sc}(s\theta,x)=-PV^s(\theta,x),\label{t7}
\end{eqnarray}
In addition, 
\begin{equation}
\Big|\tilde a_{sc}(s\theta,x)-\tilde W(s\theta,x)-\int_{-\infty}^{+\infty}F^s(\tau s\theta+x)d\tau\Big|
\le{12n^2\beta^2\big(3+{2\over 2^{-{3\over 2}}s-r}\big)^2\over \alpha(\alpha+1)(2^{-{3\over 2}}s-r)^2(1-r+{|x|\over \sqrt{2}})^{2\alpha+1}},\label{t8}
\end{equation}
\begin{eqnarray}
&&\big|\tilde b_{sc}(s\theta,x)-\int_{-\infty}^0\int_{-\infty}^\sigma F^s(\tau s\theta+x)d\tau d\sigma+\int_0^{+\infty}\int_\sigma^{+\infty} F^s(\tau s\theta+x)d\tau d\sigma\big|\nonumber\\
&&\le {24n^2\beta^2\big(3+{1\over 2^{-{3\over 2}}s-r}\big)^2
\over \alpha^2(\alpha+1)(2^{-{3\over 2}}s-r)^3(1-r+{|x|\over \sqrt{2}})^{2\alpha}},\label{t9}
\end{eqnarray}
for $r\in (0,{1\over2}+{|x|\over 2^{3\over 2}})$ and for $s>\tilde s_0(|x|,r,\beta,\alpha)$, where $\beta=\max(\beta_1^l,\beta_2,\beta_3^s)$.
\end{theorem}
Formulas \eqref{t6} and \eqref{t7} prove that $F^s$ can be reconstructed from the high energies asymptotics of the modified scattering data. For Problem \eqref{P2} this implies that $F^s$ can be reconstructed  from $\tilde a_{sc}$ and $\tilde b_{sc}$.

Estimates  \eqref{t8} and \eqref{t9} also provide the Born approximation of the modified scattering data at fixed energy, i.e. the leading term of the asymptotics of 
$\tilde a_{sc},$ $\tilde b_{sc}$, when the parameters $\alpha,$ $n,$ $s$, $\theta$ and $x$ are fixed and
$\beta$ decreases to $0$. We obtain that $F^s$ can be reconstructed from the Born approximation of the modified scattering data at fixed energy.

Estimates \eqref{t8} and \eqref{t9} also provide the first leading term in the asymptotics of the modified scattering data when the parameters $\alpha,$ $n,$ $s$, $\theta$ and $\beta$ are fixed and $|x|$ increases to $+\infty$.

\begin{proof}[Proof of Theorem \ref{thm2}]
Let  $(v_-,x_-)\in \R^n\times\R^n$, $v_-\cdot x_-=0$ and $|v_-|\ge \mu(|x_-|)$. Similarly to estimate \eqref{505a} we have 
\begin{equation}
|\eta(x_-+ \tau v_-)+(1-\eta)z_-(v_-,x_-,\tau)|
\ge {|x_-|\over \sqrt{2}}+|\tau| {|v_-|\over 2^{3\over 2}},\label{1000}
\end{equation}
for $\eta\in (0,1)$ and $\tau \in \R$.
Using \eqref{1000}, \eqref{1.4b} and \eqref{5.103b}, we obtain 
\begin{eqnarray}
&&\left|\int_{-\infty}^{+\infty}\big(F^s\big(z_-(v_-,x_-,\tau)\big)-F^s\big(\tau v_-+x_-\big)\big)d\tau\right|\nonumber\\
&\le & {2^{5\over 2}n^{3\over 2}\beta_1^l\beta_3^s\over \alpha|v_-|(1+{|x_-|\over \sqrt{2}})^\alpha}
\int_{-\infty}^{+\infty}{|\tau|\over \big(1+{|x_-|\over \sqrt{2}}+|\tau|{|v_-|\over 2\sqrt{2}}\big)^{\alpha+3}}d\tau\nonumber\\
&\le& {4n^{3\over 2}\beta_1^l\beta_3^s\over \alpha(\alpha+1)({|v_-|\over 2\sqrt{2}})^3(1+{|x_-|\over \sqrt{2}})^{2\alpha+1}},\label{5.4}
\end{eqnarray}
and
\begin{eqnarray}
&&\int_{-\infty}^0\int_{-\infty}^\sigma\left|F^s\big(z_-(v_-,x_-,\tau)\big)-F^s\big(\tau v_-+x_-\big)\right|d\tau d\sigma\nonumber\\
&\le&{2^{5\over 2}n^{3\over 2}\beta_1^l\beta_3^s\over \alpha|v_-|(1+{|x_-|\over \sqrt{2}})^\alpha}
\int_{-\infty}^0\int_{-\infty}^{\sigma}{|\tau|\over  \big(1+{|x_-|\over \sqrt{2}}+|\tau|{|v_-|\over 2\sqrt{2}} \big)^{\alpha+3}}
\label{6.5}
\end{eqnarray}
The same estimate \eqref{6.5} holds for $\int_0^{+\infty}\int_{\sigma}^{+\infty}\left|F^s\big(z_-(v_-,x_-,\tau)\big)-F^s\big(\tau v_-+x_-\big)\right|d\tau d\sigma$.
Therefore from the definition \eqref{5.2b} it follows that
\begin{eqnarray}
&&\big|\tilde l(v_-,x_-,0)-\int_{-\infty}^0\int_{-\infty}^\sigma F^s(\tau v_-+x_-)d\tau+\int_0^{+\infty}\int_\sigma^{+\infty}F^s(\tau v_-+x_-)d\tau d\sigma\big| d\sigma\nonumber\\
&&\le {4n^{3\over 2}\beta_1^l\beta_3^s\over \alpha^2(\alpha+1)\big({|v_-|\over 2\sqrt{2}}\big)^4 (1+{|x_-|\over \sqrt{2}})^{2\alpha}}.\label{6.5b}
\end{eqnarray}
Let $r>0$, $r<\max({|v_-|\over 2^{3\over 2}},{1\over 2}+{|x_-|\over \sqrt{2}})$. Note that
\begin{equation}
\max\big({\tilde \rho\over r}, \lambda,{20n\max(\beta_1^l,\beta_2)\over \alpha \big({|v_-|\over 2\sqrt{2}}-r\big)^2\big({1\over 2}+{|x_-|\over 2^{3\over 2}}-r\big)^\alpha}\big)\le
{12 \beta n\big(1+{1\over {|v_-|\over 2^{3\over 2}}-r}\big)^2\over \alpha r({|v_-|\over 2^{3\over 2}}-r)({1\over 2}-r+{|x_-|\over 2^{3\over 2}})^\alpha},\label{903b}
\end{equation}
where $\tilde \rho$ and $ \lambda$ are defined in \eqref{lb1} and \eqref{l2} respectively.
Assume that $|v_-|> \tilde s_0(|x_-|,r,\beta,\alpha)$ where $\tilde s_0$ is the root of the equation \eqref{1.30b}.
Then from \eqref{1.30b} and Lemma \ref{lem_cont2} it follows that $\A$ has a unique fixed point denoted by $y_-$ in $M_r$. Then adding \eqref{7.14c} and \eqref{5.4} we obtain \eqref{t8}. 
And adding \eqref{6.5b} and \eqref{6.44} we obtain \eqref{t9}. Theorem \ref{thm2} is proved.
\end{proof}

\subsection{Approximating the ``free" solutions $z_\pm(v,x,.)$}
\label{com2}
One may approximate the solutions $z_\pm(v,x,.)$ by the functions $z_{\pm, N+1}(v,x,.)$ defined in Section \ref{com} that are easier to compute in general. Then we can repeat the study of the previous Subsections. 
For $(v_-,x_-)\in \R^n\times\R^n$, $v_-\cdot x_-=0$, $|v_-|\ge \mu(|x_-|)$, there exists a unique solution $x(t)$ of equation \eqref{1.1} so that 
\begin{equation}
x(t)=z_{-,N+1}(v_-,x_-,t)+y_-(t),\ t\in \R\textrm{ and } \lim_{t\to-\infty}(|y_-(t)|+|\dot y_-(t)|)=0.\label{pr20}
\end{equation}
The function $y_-$ in \eqref{pr20} satisfies the following integral equation
$y_-=\A_N(y_-)$
where
\begin{equation}
\A_N(f)(t)=\int_{-\infty}^t\int_{-\infty}^\sigma\Big(F(z_{-,N+1}(v_-,x_-,\tau)+f(\tau))-F^l(z_{-,N}(v_-,x_-,\tau))\Big)d\tau d\sigma\label{pr21}
\end{equation}
for $t\in \R$ and for $f\in C(\R,\R^n)$, $\sup_{(-\infty,0]}|f|<\infty$. Then with appropriate changes in the proof of Lemma \ref{lem_cont2} we can study the operator $\A_N$ restricted to $M_r$ and we can obtain estimate and contraction estimate similar to \eqref{lb1} and \eqref{lb2}.  The decomposition \eqref{7.2} remains valid with the following changes: first we define $\tilde a_N$ by the formula \eqref{5.2ba} where we replaced $z_-(v_-,x_-,\tau)$ by $z_{-,N+1}(v_-,x_-,\tau)$; then we replace $\tilde a$ and $z_-(v_-,x_-,\tau)+y_-(\tau)$ by $\tilde a_N$ and  $z_{-,N+1}(v_-,x_-,\tau)+y_-(\tau)$ in \eqref{7.2}, \eqref{7.4}--\eqref{7.3}, and we replace $z_+$ and $F^l\big(z_-(v_-,x_-,\tau)\big)$ by $z_{+,N+1}$ and $F^l\big(z_{-,N}(v_-,x_-,\tau)\big)$ in \eqref{7.2}
and \eqref{5.2b}, and we replace $z_+$ by $z_{+,N}$ in \eqref{7.4}, \eqref{7.3}. This defines a new map $\G_{v_-,x_-}$ and  the analog of Lemma \ref{lem_cont3} can be proved. Such a result then allows to define the scattering data $\tilde b_N$. Then an analog of the  Theorem \ref{thm_y2} can be proved for the scattering solutions and scattering data $(\tilde a_N, \tilde b_N)$. Set $\tilde a_{sc,N}(v_-,x_-):=\tilde a_N(v_-,x_-)-v_-$ and $\tilde b_{sc,N}(v_-,x_-):=\tilde b_N(v_-,x_-)-x_-$. Finally the following high energies limits are valid.
Let $(\theta,x)\in T\S^{n-1}$, then 
\begin{eqnarray}
&&\lim_{s\to +\infty}s(\tilde a_{sc,N}(s\theta,x)-\tilde W_N(s\theta,x))=PF^s(\theta,x),\label{t6N}\\
&&\lim_{s\to +\infty}s^2\theta\cdot \tilde b_{sc,N}(s\theta,x)=-PV^s(\theta,x),\label{t7N}
\end{eqnarray}
where
\begin{equation}
\tilde W_N(v,x):=\int_{-\infty}^0F^l(z_{-,N}(v,x,\tau))d\tau+\int_0^{+\infty}F^l(z_{+,N}(\tilde a_N(v,x),x,\tau))d\tau.\label{t5N}
\end{equation}
for $(v,x)\in \R^n\times\R^n$, $v\cdot x=0$, $|v|>C$ for some constant $C$.
The vector $\tilde W_N$ defined by \eqref{t5N} is known from the scattering data and from $F^l$.
From \eqref{t6N} (resp. \eqref{t7N}) it follows that $F^s$ can be  reconstructed from  $\tilde a_{sc,N}$ (resp. $\tilde b_{sc,N}$). 

The limits  \eqref{t6N} and \eqref{t7N} follow from estimates similar to \eqref{t8} and \eqref{t9} that also
give the Born approximation of   $\tilde a_{sc,N},$ $\tilde b_{sc,N}$ at fixed energy, and the first leading term of the asymptotics of the scattering data
$(\tilde a_{sc,N},\tilde b_{sc,N})$ when the parameters $\alpha,$ $n,$ $v_-$ and $\beta$ are fixed and $|x_-|\to +\infty$.

\section{Proof of Lemmas \ref{lem_scatinit} and \ref{lem_decomp}}
\label{proofprel}
\begin{proof}[Proof of Lemma \ref{lem_scatinit}]
We prove the existence and uniqueness of the solution $z_+$ (similarly one can prove the existence and uniqueness of $z_-$ or just use the relation ``$z_-(w,x+h,t)=z_+(-w,x+h,-t$").
Set $C'={2^{5\over 2}n^{1\over 2}\beta_1^l\over \alpha|v|(1+{|x|\over \sqrt{2}}-|h|)^\alpha}$. Let $\V$ be the complete metric space defined by 
$$
\V:=\{g\in C(\R,\R^n)\ |\ |g(t)|\le C'|t|\textrm{ for }t\in \R\},
$$
endowed with the following norm $\|g\|_{\V}:=\sup_{t\in \R\b\{0\}}\big|{g(t)\over t}\big|$.
We consider the integral equation
\begin{equation}
G_+f(t)=-\int_0^t \int_\sigma^{+\infty} F^l(x+h+\tau w+f(\tau))d\tau d\sigma,\ t\in \R,\label{5.104}
\end{equation}
for $f\in \V$. 
First note that
\begin{eqnarray}
&&|x+h+\tau w+f(\tau)|\ge |x+\tau v|-|h|-|\tau| |v-w|-|f(\tau)|\nonumber\\
&\ge&{|x|\over \sqrt{2}}-|h|+\big({3|v|\over 4\sqrt{2}}-C'\big)|\tau|
\ge{|x|\over \sqrt{2}}-|h|+{|v|\over 2\sqrt{2}}|\tau|,\label{5.105}
\end{eqnarray}
for $\tau\in \R$ and $f\in \V$ (we used that $x\cdot v=0$ and that $C'\le {|v| \over 4\sqrt{2}}$).
Using \eqref{1.4a} we obtain that
\begin{equation}
|G_+f(t)|\le \sqrt{n}\beta_1^l\int_{-|t|}^0\int_{\sigma}^{+\infty}\big(1+{|x|\over \sqrt{2}}-|h|+{|v|\over 2\sqrt{2}}|\tau|\big)^{-\alpha-1}d\tau d\sigma
\le C'|t|,\label{5.106a}
\end{equation}
for $t\in \R$ and $f\in \V$. 
Now let $(f_1,f_2)\in \V^2$. Then using \eqref{1.4a} we have 
\begin{eqnarray*}
&&|F^l(x+h+\tau w+f_1(\tau))-F^l(x+h+\tau w+f_2(\tau))|\le n\beta_2^l|f_1-f_2|(\tau)\\
&&\times\sup_{\ep\in(0,1)}(1+|x+h+\tau w+(\ep f_1+(1-\ep) f_2)(\tau)|)^{-\alpha-2},
\end{eqnarray*}
for $\tau\in \R$. Hence using also \eqref{5.105} we have
\begin{eqnarray}
|G_+f_1(t)-G_+f_2(t)|&\le& n\beta_2^l\|f_1-f_2\|_{\V} \int_{-|t|}^0\int_\sigma^{+\infty} {|\tau| d\tau d\sigma\over \big(1+{|x|\over \sqrt{2}}-|h|+{|v|\over 2\sqrt{2}}|\tau|\big)^{\alpha+2}}\nonumber\\
&\le&{2\sqrt{2}n\beta_2^l\|f_1-f_2\|_{\V}\over |v|} \int_{-|t|}^0\int_\sigma^{+\infty} {d\tau d\sigma\over \big(1+{|x|\over \sqrt{2}}-|h|+{|v|\over 2\sqrt{2}}|\tau|\big)^{\alpha+1}}\nonumber\\
&\le&{16n\beta_2^l|t|\|f_1-f_2\|_{\V}\over \alpha|v|^2(1+{|x|\over \sqrt{2}}-|h|)^\alpha}\le 2^{-1}\|f_1-f_2\|_{\V},\label{5.106b}
\end{eqnarray}
for $t\in \R$ (we used \eqref{5.101}). From \eqref{5.106a} and \eqref{5.106b} it follows that the operator $G_+$ is a contraction map from $\V$ to $\V$. Set $z_+(w,x+h,t)=x+h+tw+f_{w,x+h}(t)$ for $t\in \R$, where $f_{w,x+h}$ denotes the unique fixed point of $G_+$ in $\V$. Then $z_+(w,x+h,.)$ satisfies \eqref{5.102}, \eqref{5.103a} and \eqref{5.103b}.
\end{proof}

Before proving Lemma \ref{lem_decomp} we recall the following standard result (see also \cite[Lemma II.2]{He}). For sake of consistency we provide a proof of Lemma \ref{lem_comp} at the end of this Section.
\begin{lemma} 
\label{lem_comp}
Let $x(t)$ be a solution of equation \eqref{1.1} and let $z(t)$ be a solution of equation \eqref{5.102}. Assume that there exists a vector $v\in \R^n$, $v\not=0$, so that
\begin{equation}
\lim_{t\to +\infty}\dot z(t)=\lim_{t\to +\infty}\dot x(t)=v.\label{l3}
\end{equation}
Then
\begin{equation}
\sup_{(0,+\infty)}|x-z|<\infty.\label{l4}
\end{equation}
\end{lemma}

\begin{proof}[Proof of Lemma \ref{lem_decomp}]
We need the following preliminary estimate \eqref{2.5b}. Using \eqref{2.3} we have
for $\tau\in \R$ and  $f\in M_r$, 
\begin{equation}
|f(\tau)|\le r|\tau|+r.\label{2.5a}
\end{equation}
Hence
\begin{eqnarray}
|z_-(v_-,\tau)+x_-+f(\tau)|&\ge& |x_-+\tau v_-|- |z_-(v_-,\tau)-\tau v_-|-|f(\tau)|\nonumber\\
&\ge&{|x_-|\over \sqrt{2}}-r+|\tau|\Big({|v_-|\over \sqrt{2}}-{2^{5\over 2}n^{1\over 2}\beta_1^l\over |v_-|\alpha}-r\Big)\nonumber\\
&\ge & {|x_-|\over \sqrt{2}}-r+|\tau|\Big({|v_-|\over 2\sqrt{2}}-r\Big).\label{2.5b}
\end{eqnarray}
We used  \eqref{5.103b} (with "$(x,w,v,h)=(0,v_-,v_-,0)$"), the inequality $|x_-+\tau v_-|\ge {|x_-|\over \sqrt{2}}+|\tau|{|v_-|\over \sqrt{2}}$ ($x_-\cdot v_-=0$) 
and \eqref{2.5a} and we used the condition
$|v_-|\ge \mu$.
Hence the integral $\int_{-\infty}^{+\infty}F(z_-(v_-,\tau)+x_-+f(\tau))d\tau$ is absolutely convergent for any $f\in M_r$. And when $y_-\in M_r$ is a fixed point for $A$ then
$z_-(v_-,.)+x_-+y_-$ satisfies equation \eqref{1.1} (see \eqref{2.1}) and 
$\dot z_-(v_-,t)+x_-+\dot y_-(t)=v_-+\int_{-\infty}^tF(z_-(v_-,\tau)+x_-+\dot y_-(\tau))d\tau\to a(v_-,x_-)\textrm{ as }t\to +\infty,$
where $a(v_-,x_-)$ is defined in \eqref{300a}. Then from Lemma \ref{lem_comp} it follows that $\sup_{t\in(0,+\infty)}|z_-(v_-,t)+x+y_-(t)-z_+(a(v_-,x_-),t)|<+\infty$. Using this latter estimate and \eqref{1.4a} and \eqref{1.4b}, and using \eqref{2.5b} we obtain that the integrals on the right hand sides of \eqref{300c} and \eqref{502c} are absolutely convergent. 
Then the decomposition \eqref{300} follows from the equality $A(y_-)=y_-$ and \eqref{2.1} and straightforward computations. 
\end{proof}

\begin{proof}[Proof of Lemma \ref{lem_comp}]
We set $\delta(t)=x(t)-z(t)$ for $t\ge 0$. Property \eqref{l3} shows that there exists $\ep>0$ so that
\begin{equation}
1+|\eta x(t)+(1-\eta)z(t)|\ge \ep(1+t),\textrm{ for }t\ge 0 \textrm{ and }\eta\in [0,1].\label{402}
\end{equation}
Then from equations \eqref{1.1} and \eqref{5.102} it follows that
\begin{equation}
\delta(t)
=\delta(0)-\int_0^t\int_\sigma^{+\infty}F^s(x(\tau))d\tau d\sigma-\int_0^t\int_\sigma^{+\infty}(F^l(x(\tau))-F^l(z(\tau)))d\tau d\sigma,\label{401}
\end{equation}
for $t\ge 0$, where the integrals on the right hand side of \eqref{401} are absolutely convergent (see \eqref{402} and \eqref{1.4a} and \eqref{1.4b}).
Note that
\begin{equation}
\int_0^t\int_\sigma^{+\infty}|F^s(x(\tau))|d\tau d\sigma\le \sqrt{n}\beta_2\ep^{-\alpha-2}\int_0^t\int_\sigma^{+\infty}{d\tau d\sigma\over (1+\tau)^{\alpha+2}}
\le{\sqrt{n}\beta_2\over \alpha(\alpha+1)\ep^{\alpha+2}},\label{403}
\end{equation}
for $t\ge 0$. Hence using \eqref{401} we obtain
\begin{equation}
|\delta(t)|
\le C_0+\int_0^t\int_\sigma^{+\infty}|F^l(x(\tau))-F^l(z(\tau))|d\tau d\sigma,\textrm{ for }t\ge0,\label{404}
\end{equation}
where $C_0=|\delta(0)|+{\sqrt{n}\beta_2\over \alpha(\alpha+1)\ep^{\alpha+2}}$.

One may assume without loss of generality that $\alpha\not={1\over m}$ for any $m\in \N$. Otherwise replace $\alpha$ by some $\alpha'\in (0,\alpha)$ so that $\alpha'
\not={1\over m'}$ for any $m'\in \N$.
Then
\begin{eqnarray}
\int_0^t\int_\sigma^{+\infty}\Big(|F^l(x(\tau))|+|F^l(z(\tau)|\Big)d\tau d\sigma
&\le& 2\sqrt{n}\beta_1^l\ep^{-\alpha-1}\int_0^t\int_\sigma^{+\infty}{d\tau d\sigma\over (1+\tau)^{\alpha+1}}\nonumber\\
&\le& {2\sqrt{n}\beta_1^l\over \alpha(1-\alpha)\ep^{\alpha+1}}(1+t)^{1-\alpha},
\label{406}
\end{eqnarray}
for $t\ge 0$. 
Using also \eqref{404}  we obtain that there exist positive constants $C_1$ and $C_1'$ so that $|\delta(t)|
\le C_1+C_1't^{1-\alpha}$ for $t\ge 0$.
Now using \eqref{404}, the growth properties of $F^l$ \eqref{1.4a} and \eqref{402} we obtain
\begin{eqnarray}
|\delta(t)|
&\le& C_0+\int_0^t\int_\sigma^{+\infty}\sup_{\eta\in (0,1)}{n\beta_2|x(\tau)-z(\tau)|\over (1+|\eta x(\tau)+(1-\eta)z(\tau)|)^{\alpha+2}}d\tau d\sigma\nonumber\\
&\le&C_0+n\beta_2\ep^{-\alpha-2}\int_0^t\int_\sigma^{+\infty}{|\delta(\tau)|\over (1+\tau)^{\alpha+2}}d\tau d\sigma,\label{404b}
\end{eqnarray}
for $t\ge 0$. 
Then using  \eqref{404b} we prove by induction the following:
For any $m=1\ldots\lfloor \alpha^{-1}\rfloor$ there exist positive constants $C_m$ and $C_m'$ so that
$|\delta(t)|\le C_m+C'_mt^{1-m\alpha}$ for $t\ge 0$. Combining again this latter estimate for $m=\lfloor \alpha^{-1}\rfloor$ and  the estimate \eqref{404b} we obtain
\begin{eqnarray*}
|\delta(t)|
&\le&C_0+n\beta_2\ep^{-\alpha-2}\int_0^t\int_\sigma^{+\infty}{C_m+C'_m\tau^{1-m\alpha}\over (1+\tau)^{\alpha+2}}d\tau d\sigma\nonumber\\
&\le&C_0+{n\beta_2C_m\over\alpha(\alpha+1)\ep^{\alpha+2}}+{C'_mn\beta_2\over (m+1)\alpha((m+1)\alpha-1)\ep^{\alpha+2}},
\end{eqnarray*}
for $t\ge 0$ ($(m+1)\alpha-1>0$), which proves the lemma.
\end{proof}

\section{Proof of Lemmas \ref{lem_cont} and \ref{lem_cont2}}
\label{proof_cont}
\begin{proof}[Proof of Lemma \ref{lem_cont}]
We first prove \eqref{l1}. We need the following  estimates for $A(f)(t)$ \eqref{2.11} and \eqref{2.19}.
Using \eqref{1.4a}, \eqref{1.4b} and \eqref{2.5b} we obtain 
\begin{equation}
|F^s(z_-(v_-,\tau)+x_-+f(\tau))|\le{\beta_2 \sqrt{n}\over \big(1+{|x_-|\over \sqrt{2}}-r+|\tau|\big({|v_-|\over 2\sqrt{2}}-r\big)\big)^{\alpha+2}},\label{2.7b}
\end{equation}
and 
\begin{equation}
|F^l(z_-(v_-,\tau)+x_-+f(\tau))-F^l(z_-(v_-,\tau))|\le{\beta_2 n(|x_-|+|f(\tau)|)
\over \big(1-r+|\tau|\big({|v_-|\over 2\sqrt{2}}-r\big)\big)^{\alpha+2}},\label{2.7c}
\end{equation}
for $\tau\in \R$.
In addition from \eqref{2.1} it follows that
\begin{equation}
\dot A(f)(t)=\int_{-\infty}^t\Big(F(z_-(v_-,\tau)+x_-+f(\tau))-F^l(z_-(v_-,\tau))\Big)d\tau,\ \tau\in \R.\label{2.7d}
\end{equation}
Therefore combining \eqref{2.7b}, \eqref{2.7c} (with $|f(\tau)|\le \sup_{(-\infty,0)}|f|$ for $\tau\le 0$) and \eqref{2.7d} we obtain
\begin{eqnarray}
|\dot A(f)(t)|&\le &\beta_2\int_{-\infty}^t{\sqrt{n}+n(|x_-|+|f(\tau)|)\over\big(1-r+|\tau|\big({|v_-|\over 2\sqrt{2}}-r\big)\big)^{\alpha+2}}d\tau\nonumber\\
&\le&{\beta_2\big(n(|x_-|+\sup_{(-\infty,0)}|f|)+\sqrt{n}\big)\over (\alpha+1)\big({|v_-|\over 2\sqrt{2}}-r\big)\Big(1-r+|t|\big({|v_-|\over 2\sqrt{2}}-r\big)\Big)^{\alpha+1}},\label{2.10}
\end{eqnarray}
for $t\le 0$. Then integrating over $(-\infty,t)$ we obtain
\begin{equation}
|A(f)(t)|\le {\beta_2(n(|x_-|+\sup_{(-\infty,0)}|f|)+\sqrt{n})\over \alpha(\alpha+1)\big({|v_-|\over 2\sqrt{2}}-r\big)^2\Big(1-r+|t|\big({|v_-|\over 2\sqrt{2}}-r\big)\Big)^\alpha},\label{2.11}
\end{equation}
for $t\le 0$.

Now let $t\ge 0$.
Using \eqref{2.7d} we have 
\begin{eqnarray}
\dot A(f)(t)&=&\dot A(f)(0)+\int_0^t\big(F^l(z_-(v_-,\tau)+x_-+f(\tau))-F^l(z_-(v_-,\tau))\big)d\tau\nonumber\\
&&+\int_0^tF^s(z_-(v_-,\tau)+x_-+f(\tau))d\tau.\label{2.12}
\end{eqnarray}
Hence from \eqref{2.10} (with "$t=0$"), \eqref{2.7b} and \eqref{2.7c} (with $|f(\tau)|\le (1+\tau)\sup_{s\in (0+\infty)}{|f(s)|\over 1+s}$ for $\tau\ge 0$) it follows that
\begin{eqnarray}
&&|\dot A(f)(t)|\le{\beta_2\big(n(|x_-|+\sup_{(-\infty,0)}|f|)+\sqrt{n}\big)\over (\alpha+1)\big({|v_-|\over 2\sqrt{2}}-r\big)\big(1-r\big)^{\alpha+1}}\nonumber\\
&&+\beta_2\int_0^t{n^{1\over 2}+n|x_-|+n(1+\tau)\sup_{s\in (0,+\infty)}{|f(s)|\over 1+s}\over \Big(1-r+|\tau|\big({|v_-|\over 2\sqrt{2}}-r\big)\Big)^{\alpha+2}}d\tau\nonumber\\
&&\le {\beta_2\big(2n|x_-|+n\|f\|+2\sqrt{n}\big)\over (\alpha+1)\big({|v_-|\over 2\sqrt{2}}-r\big)\Big(1-r\Big)^{\alpha+1}}
+{n\beta_2\sup_{s\in (0,+\infty)}{|f(s)|\over 1+s}\over \alpha\big({|v_-|\over 2\sqrt{2}}-r\big)^2(1-r)^\alpha}.\label{2.15}
\end{eqnarray}

We also have 
\begin{eqnarray}
A(f)(t)&=&A(f)(0)+t\dot A(f)(t)\nonumber\\
&&-\int_0^t\int_\sigma^t\Big(F^l(z_-(v_-,\tau)+x_-+f(\tau))-F^l(z_-(v_-,\tau))\Big)d\tau d\sigma\nonumber\\
&&-\int_0^t\int_\sigma^t F^s(z_-(v_-,\tau)+x_-+f(\tau))d\tau d\sigma.\label{2.16}
\end{eqnarray}
Using \eqref{2.7b} and \eqref{2.7c} we obtain
\begin{equation}
\left|\int_0^t\int_\sigma^t F^s(z_-(v_-,\tau)+x_-+f(\tau))d\tau d\sigma\right|
\le{\beta_2\sqrt{n}\over \alpha(\alpha+1)\big({|v_-|\over 2\sqrt{2}}-r\big)^2\Big(1+{|x_-|\over \sqrt{2}}-r\Big)^{\alpha}},\label{2.17}
\end{equation}
\begin{eqnarray}
&&\left|\int_0^t\int_\sigma^t\big(F^l(z_-(v_-,\tau)+x_-+f(\tau))-F^l(z_-(v_-,\tau))\big)d\tau d\sigma\right|\nonumber\\
&\le& t\int_0^{+\infty}\big|F^l(z_-(v_-,\tau)+x_-+f(\tau))-F^l(z_-(v_-,\tau))\big|d\tau\nonumber\\
&\le&
 t\Big({n\beta_2\sup_{s\in (0,+\infty)}{|f(s)|\over 1+s}\over \alpha\big({|v_-|\over 2\sqrt{2}}-r\big)^2(1-r)^\alpha}
+ {n\beta_2\big(|x_-|+\sup_{s\in (0,+\infty)}{|f(s)|\over 1+s}\big)\over(\alpha+1)\big({|v_-|\over 2\sqrt{2}}-r\big)(1-r)^{\alpha+1}}\Big).\label{2.18}
\end{eqnarray}
Combining \eqref{2.16}, \eqref{2.11} (with "$t=0$"), \eqref{2.15}, \eqref{2.17} and \eqref{2.18} we obtain
\begin{eqnarray}
&&|A(f)(t)|\le {\beta_2(n(|x_-|+\sup_{(-\infty,0)}|f|)+2\sqrt{n})\over \alpha(\alpha+1)\big({|v_-|\over 2\sqrt{2}}-r\big)^2\Big(1-r\Big)^\alpha}\nonumber\\
&&+t\beta_2\Big({\big(3n|x_-|+2n\|f\|+2\sqrt{n}\big)\over (\alpha+1)\big({|v_-|\over 2\sqrt{2}}-r\big)\Big(1-r\Big)^{\alpha+1}}
+{2n\sup_{s\in (0,+\infty)}{|f(s)|\over 1+s}\over \alpha\big({|v_-|\over 2\sqrt{2}}-r\big)^2(1-r)^\alpha}\Big).\label{2.19}
\end{eqnarray}
Then \eqref{l1} follows from \eqref{2.11} and \eqref{2.19} and the estimate $\|f\|\le r$.

It remains to prove \eqref{l2}. Estimate \eqref{l2} will follow from \eqref{2.22} and \eqref{2.30} given below.
Let $(f_1,f_2)\in M_r^2$. 
Using \eqref{1.4a}, \eqref{1.4b} and \eqref{2.5b} we obtain 
\begin{equation}
|F^l(z_-(v_-,\tau)+x_-+f_1(\tau))-F^l(z_-(v_-,\tau)+x_-+f_2(\tau))|\le{\beta_2 n|f_1-f_2|(\tau)
\over \big(1+{|x_-|\over \sqrt{2}}-r+|\tau|\big({|v_-|\over 2\sqrt{2}}-r\big)\big)^{\alpha+2}},\label{2.20a}
\end{equation}
\begin{equation}
|F^s(z_-(v_-,\tau)+x_-+f_1(\tau))-F^s(z_-(v_-,\tau)+x_-+f_2(\tau))|\le{\beta_3^s n|f_1-f_2|(\tau)
\over \big(1+{|x_-|\over \sqrt{2}}-r+|\tau|\big({|v_-|\over 2\sqrt{2}}-r\big)\big)^{\alpha+3}},\label{2.20b}
\end{equation}
for $\tau\in \R$.
In addition from \eqref{2.7d} it follows that
\begin{eqnarray}
&&\dot A(f_1)(t)-\dot A(f_2)(t)\nonumber\\
&&=\int_{-\infty}^t\big(F^l(z_-(v_-,\tau)+x_-+f_1(\tau))-F^l(z_-(v_-,\tau)+x_-+f_2(\tau))\big)d\tau\nonumber
\end{eqnarray}
\begin{equation}
+\int_{-\infty}^t\Big(F^s(z_-(v_-,\tau)+x_-+f_1(\tau))-F^s(z_-(v_-,\tau)+x_-+f_2(\tau))\Big)d\tau.\label{2.20}
\end{equation}
Hence we integrate in the $\tau$ variable over the interval $(-\infty,t)$, $t\le 0$, both sides of \eqref{2.20a} and \eqref{2.20b} where we use the inequality $|f_1-f_2|(\tau)\le \sup_{(-\infty,0)}|f_1-f_2|$, and we obtain
\begin{eqnarray}
|\dot A(f_1)(t)-\dot A(f_2)(t)|
&\le& {n\sup_{(-\infty,0]}|f_1-f_2|\over \Big({|v_-|\over 2\sqrt{2}}-r\Big)\Big(1+{|x_-|\over \sqrt{2}}-r+|t|\Big({|v_-|\over 2\sqrt{2}}-r\Big)\Big)^{\alpha+1}}\left[ {\beta_2\over (\alpha+1)}\right.\nonumber\\
&&+\left. {\beta_3^s\over (\alpha+2)\Big(1+{|x_-|\over \sqrt{2}}-r+|t|\Big({|v_-|\over 2\sqrt{2}}-r\Big)\Big)}\right],
\label{2.21}
\end{eqnarray}
for $t\le 0$. Then we integrate in the $t$ variable both sides of \eqref{2.21}  and we obtain
\begin{eqnarray}
|A(f_1)(t)-A(f_2)(t)|
&\le& {n\sup_{(-\infty,0]}|f_1-f_2|\over(\alpha+1) \Big({|v_-|\over 2\sqrt{2}}-r\Big)^2\Big(1+{|x_-|\over \sqrt{2}}-r+|t|\Big({|v_-|\over 2\sqrt{2}}-r\Big)\Big)^\alpha}\left[ {\beta_2\over \alpha}\right.\nonumber\\
&&+\left. {\beta_3^s\over (\alpha+2)\Big(1+{|x_-|\over \sqrt{2}}-r+|t|\Big({|v_-|\over 2\sqrt{2}}-r\Big)\Big)}\right],
\label{2.22}
\end{eqnarray}
for $t\le 0$.

From \eqref{2.12} it follows that
\begin{eqnarray}
&&\dot A(f_1)(t)-\dot A(f_2)(t)=\dot A(f_1)(0)-\dot A(f_2)(0)\label{2.23}\\
&&+\int_0^t\big(F^l(z_-(v_-,\tau)+x_-+f_1(\tau))-F^l(z_-(v_-,\tau)+x_-+f_2(\tau))\big)d\tau\nonumber\\
&&+\int_0^t\Big(F^s(z_-(v_-,\tau)+x_-+f_1(\tau))-F^s(z_-(v_-,\tau)+x_-+f_2(\tau))\Big)d\tau.\nonumber
\end{eqnarray}
We have 
\begin{eqnarray}
&&\Big|\int_0^t\Big(F^s(z_-(v_-,\tau)+x_-+f_1(\tau))-F^s(z_-(v_-,\tau)+x_-+f_2(\tau))\Big)d\tau\Big| \nonumber\\
&&\le n\beta_3^s\sup_{s\in [0,+\infty)}{|(f_1-f_2)(s)|\over 1+s}\int_0^t{1+\tau\over \Big(1+{|x_-|\over \sqrt{2}}-r+|\tau|\Big({|v_-|\over 2\sqrt{2}}-r\Big)\Big)^{\alpha+3}}d\tau\nonumber
\end{eqnarray}
\begin{equation}
\le {n\beta_3^s\sup_{s\in [0,+\infty)}{|(f_1-f_2)(s)|\over 1+s}\over\Big({|v_-|\over 2\sqrt{2}}-r\Big)(1+{|x_-|\over \sqrt{2}}-r)^{\alpha+1}}\left[{1\over (\alpha+2) (1+{|x_-|\over \sqrt{2}}-r)}+{1\over (\alpha+1)\Big({|v_-|\over 2\sqrt{2}}-r\Big)}\right],\label{2.24}
\end{equation}
and 
\begin{equation*}
\Big|\int_0^t\Big(F^l(z_-(v_-,\tau)+x_-+f_1(\tau))-F^l(z_-(v_-,\tau)+x_-+f_2(\tau))\Big)d\tau\Big| 
\end{equation*}
\begin{equation}
\le {n\beta_2\sup_{s\in [0,+\infty)}{|(f_1-f_2)(s)|\over 1+s}\over\Big({|v_-|\over 2\sqrt{2}}-r\Big)(1+{|x_-|\over \sqrt{2}}-r)^{\alpha}}\left[{1\over (\alpha+1)(1+{|x_-|\over \sqrt{2}}-r)}+{1\over \alpha\Big({|v_-|\over 2\sqrt{2}}-r\Big)}\right].\label{2.25}
\end{equation}
Hence using also \eqref{2.21} (for ``$t=0$") we obtain
\begin{eqnarray}
&&|\dot A(f_1)(t)-\dot A(f_2)(t)|
\le {n\|f_1-f_2\|\over \Big({|v_-|\over 2\sqrt{2}}-r\Big)\Big(1+{|x_-|\over \sqrt{2}}-r\Big)^{\alpha+1}}\left[ {\beta_2\over (\alpha+1)}\right.\nonumber\\
&&+\left. {\beta_3^s\over (\alpha+2)\Big(1+{|x_-|\over \sqrt{2}}-r\Big)}\right]\nonumber\\
&&+ {n\sup_{s\in [0,+\infty)}{|(f_1-f_2)(s)|\over 1+s}\over \Big({|v_-|\over 2\sqrt{2}}-r\Big)^2(1+{|x_-|\over \sqrt{2}}-r)^{\alpha}}\left[{\beta_3^s\over(\alpha+1) (1+{|x_-|\over \sqrt{2}}-r)}+{\beta_2\over\alpha}\right].\label{2.26}
\end{eqnarray}
Next we consider
\begin{eqnarray}
&&A(f_1)(t)-A(f_2)(t)=A(f_1)(0)-A(f_2)(0)+t(\dot A(f_1)(t)-\dot A(f_2)(t))\label{2.27}\\
&&+\int_0^t\int_\sigma^t\Big(F^l(z_-(v_-,\tau)+x_-+f_2(\tau))-F^l(z_-(v_-,\tau)+x_-+f_1(\tau))\Big)d\tau d\sigma\nonumber\\
&&-\int_0^t\int_\sigma^t\Big(F^s(z_-(v_-,\tau)+x_-+f_1(\tau))-F^s(z_-(v_-,\tau)+x_-+f_2(\tau))\Big)d\tau d\sigma\nonumber
\end{eqnarray}
Using \eqref{2.20a} we obtain
\begin{equation}
\int_0^t\int_\sigma^t\Big|F^s(z_-(v_-,\tau)+x_-+f_1(\tau))-F^s(z_-(v_-,\tau)+x_-+f_2(\tau))\Big|d\tau d\sigma\nonumber
\end{equation}
\begin{equation}
\le {n\beta_3^s\sup_{s\in [0,+\infty)}{|(f_1-f_2)(s)|\over 1+s}\over\alpha(\alpha+1)\big({|v_-|\over 2\sqrt{2}}-r\big)^2(1+{|x_-|\over \sqrt{2}}-r)^\alpha}\left[{1\over \big({|v_-|\over 2\sqrt{2}}-r\big)}+{\alpha\over (\alpha+2)(1+{|x_-|\over \sqrt{2}}-r)}\right],\label{2.28}
\end{equation}
and 
\begin{eqnarray}
&&\int_0^t\int_\sigma^t\Big|F^l(z_-(v_-,\tau)+x_-+f_2(\tau))-F^l(z_-(v_-,\tau)+x_-+f_1(\tau))\Big|d\tau d\sigma \nonumber\\
&&\le {n\beta_2 \sup_{s\in [0,+\infty)}{|(f_1-f_2)(s)|\over 1+s}\over\alpha\big({|v_-|\over 2\sqrt{2}}-r\big)^2(1+{|x_-|\over \sqrt{2}}-r)^{\alpha}}\big({1\over (\alpha+1)}+t\big).\label{2.29}
\end{eqnarray}
Therefore using also \eqref{2.22} (with "$t=0$"), \eqref{2.26} and \eqref{2.27}, we have
\begin{eqnarray}
&&|A(f_1)(t)-A(f_2)(t)|\le{ n\|f_1-f_2\|\over(\alpha+1) \Big({|v_-|\over 2\sqrt{2}}-r\Big)^2\Big(1+{|x_-|\over \sqrt{2}}-r\Big)^\alpha}\left[ {\beta_2\over \alpha}+{\beta_3^s\over (\alpha+2)\Big(1+{|x_-|\over \sqrt{2}}-r\Big)}\right]\nonumber\\
&&+{n\beta_3^s\sup_{s\in [0,+\infty)}{|(f_1-f_2)(s)|\over 1+s}\over\alpha(\alpha+1)\big({|v_-|\over 2\sqrt{2}}-r\big)^3(1+{|x_-|\over \sqrt{2}}-r)^\alpha}\nonumber\\
&&+t{n\|f_1-f_2\|\over \Big({|v_-|\over 2\sqrt{2}}-r\Big)\Big(1+{|x_-|\over \sqrt{2}}-r\Big)^{\alpha+1}}\left[ {\beta_2\over (\alpha+1)}
+{\beta_3^s\over (\alpha+2)\Big(1+{|x_-|\over \sqrt{2}}-r\Big)}\right]\nonumber\\
&&+ t{n\sup_{s\in [0,+\infty)}{|(f_1-f_2)(s)|\over 1+s}\over \Big({|v_-|\over 2\sqrt{2}}-r\Big)^2(1+{|x_-|\over \sqrt{2}}-r)^{\alpha}}\left[{\beta_3^s\over(\alpha+1) (1+{|x_-|\over \sqrt{2}}-r)}+{2\beta_2\over\alpha}\right].\label{2.30}
\end{eqnarray}

\end{proof}

\begin{proof}[Proof of Lemma \ref{lem_cont2}]
We follow the proof of Lemma \ref{lem_cont}.

Estimate \eqref{2.5a} for $\tau\in \R$, $f\in M_r$ still holds and we have
\begin{eqnarray}
|z_-(v_-,x_-,\tau)+f(\tau)|&\ge& {|x_-|\over \sqrt{2}}-r+|\tau|\Big({|v_-|\over \sqrt{2}}-{2^{5\over 2}n^{1\over 2}\beta_1^l\over |v_-|\alpha
(1+{|x_-|\over \sqrt{2}})^\alpha}-r\Big)\nonumber\\
&\ge & {|x_-|\over \sqrt{2}}-r+|\tau|\Big({|v_-|\over 2\sqrt{2}}-r\Big).\label{4.5b}
\end{eqnarray}
We used \eqref{5.103b} (with "$(x,w,h)=(x_-,v_-,0)$") and the inequality 
$|v_-|\ge \mu(|x_-|)$.
Then using \eqref{1.4a}, \eqref{1.4b} and \eqref{4.5b} we obtain 
\begin{equation}
|F^s(z_-(v_-,x_-,\tau)+f(\tau))|\le{\beta_2 \sqrt{n}\over \big(1+{|x_-|\over \sqrt{2}}-r+|\tau|\big({|v_-|\over 2\sqrt{2}}-r\big)\big)^{\alpha+2}},\label{4.7b}
\end{equation}
and 
\begin{equation}
|F^l(z_-(v_-,x_-,\tau)+f(\tau))-F^l(z_-(v_-,x_-,\tau))|\le{\beta_2 n|f(\tau)|
\over \big(1+{|x_-|\over \sqrt{2}}-r+|\tau|\big({|v_-|\over 2\sqrt{2}}-r\big)\big)^{\alpha+2}},\label{4.7c}
\end{equation}
for $\tau\in \R$. 
Then the proof of the following estimates \eqref{4.10}, \eqref{4.11}, \eqref{4.15}, \eqref{4.19} is similar to the proof of the estimates \eqref{2.10}, \eqref{2.11}, \eqref{2.15} and \eqref{2.19} respectively, and we have 
\begin{equation}
|\dot \A(f)(t)|\le {\beta_2\big(n\sup_{(-\infty,0)}|f|+\sqrt{n}\big)\over (\alpha+1)\big({|v_-|\over 2\sqrt{2}}-r\big)\Big(1+{|x_-|\over \sqrt{2}}-r+|t|\big({|v_-|\over 2\sqrt{2}}-r\big)\Big)^{\alpha+1}},\label{4.10}
\end{equation}
for $t\le 0$,
\begin{equation}
|\A(f)(t)|\le {\beta_2(n\sup_{(-\infty,0)}|f|+\sqrt{n})\over \alpha(\alpha+1)\big({|v_-|\over 2\sqrt{2}}-r\big)^2\Big(1+{|x_-|\over \sqrt{2}}-r+|t|\big({|v_-|\over 2\sqrt{2}}-r\big)\Big)^\alpha},\label{4.11}
\end{equation}
for $t\le 0$, 
\begin{equation}
|\dot \A(f)(t)|\le{\beta_2\big(n\|f\|+2\sqrt{n}\big)\over (\alpha+1)\big({|v_-|\over 2\sqrt{2}}-r\big)\Big(1+{|x_-|\over \sqrt{2}}-r\Big)^{\alpha+1}}
+{n\beta_2\sup_{s\in (0,+\infty)}{|f(s)|\over 1+s}\over \alpha\big({|v_-|\over 2\sqrt{2}}-r\big)^2(1+{|x_-|\over \sqrt{2}}-r)^\alpha},\label{4.15}
\end{equation}
for $t\ge 0$, and
\begin{eqnarray}
&&|\A(f)(t)|\le {\beta_2(n(\sup_{(-\infty,0)}|f|)+2\sqrt{n})\over \alpha(\alpha+1)\big({|v_-|\over 2\sqrt{2}}-r\big)^2\Big(1+{|x_-|\over \sqrt{2}}-r\Big)^\alpha}\label{4.19}\\
&&+t\beta_2\Big({\big(2n\|f\|+2\sqrt{n}\big)\over (\alpha+1)\big({|v_-|\over 2\sqrt{2}}-r\big)\Big(1+{|x_-|\over \sqrt{2}}-r\Big)^{\alpha+1}}
+{2n\sup_{s\in (0,+\infty)}{|f(s)|\over 1+s}\over \alpha\big({|v_-|\over 2\sqrt{2}}-r\big)^2(1+{|x_-|\over \sqrt{2}}-r)^\alpha}\Big),\nonumber
\end{eqnarray}
for $t\ge 0$.
Estimate \eqref{lb1} follows from \eqref{4.11} and \eqref{4.19}.

Let $(f_1,f_2)\in M_r^2$. 
Using \eqref{1.4a}, \eqref{1.4b} and \eqref{4.5b} we obtain 
\begin{equation}
|F^l(z_-(v_-,x_-,\tau)+f_1(\tau))-F^l(z_-(v_-,x_-,\tau)+f_2(\tau))|\le{\beta_2 n|f_1-f_2|(\tau)
\over \big(1+{|x_-|\over \sqrt{2}}-r+|\tau|\big({|v_-|\over 2\sqrt{2}}-r\big)\big)^{\alpha+2}},\label{4.20a}
\end{equation}
\begin{equation}
|F^s(z_-(v_-,x_-,\tau)+f_1(\tau))-F^s(z_-(v_-,x_-,\tau)+f_2(\tau))|\le{\beta_3^s n|f_1-f_2|(\tau)
\over \big(1+{|x_-|\over \sqrt{2}}-r+|\tau|\big({|v_-|\over 2\sqrt{2}}-r\big)\big)^{\alpha+3}},\label{4.20b}
\end{equation}
for $\tau\in \R$.
Then similarly to the proof of \eqref{2.21} (resp.  \eqref{2.22}) and  \eqref{2.26} (resp. \eqref{2.30}) we prove the following: $\dot \A(f_1)(t)-\dot \A(f_2)(t)$ (resp. $\A(f_1)(t)-\A(f_2)(t)$) for $t\le 0$ is bounded by the right-hand side of \eqref{2.21} (resp. \eqref{2.22}); $\dot \A(f_1)(t)-\dot \A(f_2)(t)$ (resp. $\A(f_1)(t)-\A(f_2)(t)$) for $t\ge 0$ is bounded by the right-hand side of \eqref{2.26} (resp. \eqref{2.30}).
Then estimate \eqref{lb2} follows from these latter bounds on $\A(f_1)(t)-\A(f_2)(t)$, $t\in \R$.

It remains to prove \eqref{5.9}. Let $f\in M_r$.
Using \eqref{1.4a} we obtain 
\begin{equation}
|F^l(z_-(v_-,x_-,\tau)+f(\tau))|\le{\beta_1^l \sqrt{n}\over \big(1+{|x_-|\over \sqrt{2}}-r+|\tau|\big({|v_-|\over 2\sqrt{2}}-r\big)\big)^{\alpha+1}},\label{4.7bb}
\end{equation}
for $\tau\in \R$.
Then from \eqref{5.2a}, \eqref{4.7b} and \eqref{4.7bb} it follows that 
\begin{eqnarray*}
|\tilde k(v_-,x_-,f)-v_-|&\le& \int_{-\infty}^{+\infty}{\sqrt{n}\beta_1^l d\tau\over  \big(1+{|x_-|\over \sqrt{2}}-r+|\tau|\big({|v_-|\over 2\sqrt{2}}-r\big)\big)^{\alpha+1}}\\
&&+ \int_{-\infty}^{+\infty}{\sqrt{n}\beta_2 d\tau\over\big(1+{|x_-|\over \sqrt{2}}-r+|\tau|\big({|v_-|\over 2\sqrt{2}}-r\big)\big)^{\alpha+2}},
\end{eqnarray*}
which proves \eqref{5.9}.
\end{proof}

\section{Proof of Theorem \ref{thm_y}}
\label{proof_thm_y}
The estimate \eqref{t10a} (resp. \eqref{t10b}) follows from the assumption $y_-=A(y_-)$, and $\sup_{(-\infty,0)}|y_-|\le r$ and the estimate \eqref{2.10} (resp. \eqref{2.11}).
Using \eqref{1.4b} and \eqref{2.5b} we obtain 
\begin{equation}
|F^l(z_-(v_-,\tau)+x_-+y_-(\tau))|\le{\beta_1^l \sqrt{n}\over \big(1+{|x_-|\over \sqrt{2}}-r+|\tau|\big({|v_-|\over 2\sqrt{2}}-r\big)\big)^{\alpha+1}},\label{2.7bbb}
\end{equation}
for $\tau \in \R$.
Using \eqref{300a}, \eqref{2.7b} and \eqref{2.7bbb} we obtain
\begin{eqnarray*}
|a_{sc}(v_-,x_-)|&\le& \int_{-\infty}^{+\infty}{\beta_1^l n^{1\over 2}\over (1+{|x_-|\over \sqrt{2}}-r+\big({|v_-|\over 2\sqrt{2}}-r\big)|\tau|)^{\alpha+1}}d\tau\nonumber\\
&&+ \int_{-\infty}^{+\infty}{\beta_2 n^{1\over 2}\over (1+{|x_-|\over \sqrt{2}}-r+\big({|v_-|\over 2\sqrt{2}}-r\big)|\tau|)^{\alpha+2}}d\tau,
\end{eqnarray*}
which proves \eqref{503a}.

Then using \eqref{502a}, \eqref{2.7b} and \eqref{2.7c} we have 
\begin{eqnarray*}
|l(v_-,x_-,y_-)|&\le& 2\beta_2n^{1\over 2}\int_{-\infty}^0\int_{-\infty}^\sigma \big(1+{|x_-|\over \sqrt{2}}-r+\big({|v_-|\over 2 \sqrt{2}}-r\big)|\tau|\big)^{-\alpha-2}d\tau d\sigma\nonumber\\
&&+\beta_2 n\int_{-\infty}^0\int_{-\infty}^\sigma {|x_-|+\sup_{(-\infty,0)}|y_-|\over \big(1-r+\big({|v_-|\over 2 \sqrt{2}}-r\big)|\tau|\big)^{\alpha+2}}d\tau d\sigma\nonumber\\
&\le &{2\beta_2 n^{1\over 2}\over \alpha (\alpha+1)\big({|v_-|\over 2 \sqrt{2}}-r\big)^2 (1+{|x_-|\over \sqrt{2}}-r\big)^\alpha}
+{\beta_2n(|x_-|+\sup_{(-\infty,0)}|y_-|)\over \alpha (\alpha+1)\big({|v_-|\over 2 \sqrt{2}}-r\big)^2(1-r)^\alpha}\nonumber\\
&\le&{\beta_2 (n|x_-|+2n^{1\over 2}+n\sup_{(-\infty,0)}|y_-|)\over \alpha (\alpha+1)(1-r)^\alpha\big({|v_-|\over 2 \sqrt{2}}-r\big)^2}.\nonumber
\end{eqnarray*}
This estimate with the estimate $\sup_{(-\infty,0)}|y_-|\le r$ proves \eqref{503b}.
Now using \eqref{300a} and \eqref{2.26} we obtain that
\begin{eqnarray}
&&|a_{sc}(v_-,x_-)-\int_{-\infty}^{+\infty}F(z_-(v_-,\tau)+x_-)|=\lim_{t\to+\infty}|\dot A(y_-)(t)-\dot A(0)(t)|\nonumber\\
&&\le {n\|y_-\|\over \Big({|v_-|\over 2\sqrt{2}}-r\Big)\Big(1+{|x_-|\over \sqrt{2}}-r\Big)^{\alpha+1}}\left[ {\beta_2\over (\alpha+1)}
+ {\beta_3^s\over (\alpha+2)\Big(1+{|x_-|\over \sqrt{2}}-r\Big)}\right]\nonumber\\
&&+ {n\sup_{s\in [0,+\infty)}{|y_-(s)|\over 1+s}\over \Big({|v_-|\over 2\sqrt{2}}-r\Big)^2(1+{|x_-|\over \sqrt{2}}-r)^{\alpha}}\left[{\beta_3^s\over(\alpha+1) (1+{|x_-|\over \sqrt{2}}-r)}+{\beta_2\over\alpha}\right].\label{3.3}
\end{eqnarray}
Then note that $\|y_-\|=\|A(y_-)\|$ is bounded by the right hand side of \eqref{l1}. Hence combining this latter bound on $\|y_-\|$  and \eqref{3.3} and the estimate $1+{|x_-|\over \sqrt{2}}-r\ge 1-r$ we obtain \eqref{504a}.

Using \eqref{502a}, \eqref{2.20a} and \eqref{2.20b} we obtain
\begin{eqnarray}
&&|l(v_-,x_-,y_-)-l(v_-,x_-,0)|\nonumber\\
&\le& \int_{-\infty}^0\int_{-\infty}^\sigma \left|F^l\big(z_-(v_-,\tau)+x_-+y_-(\tau)\big)-F^l\big(z_-(v_-,\tau)+x_-\big)\right|d\tau d\sigma\nonumber\\
&&+ \int_{-\infty}^0\int_{-\infty}^\sigma \left|F^s\big(z_-(v_-,\tau)+x_-+y_-(\tau)\big)-F^s\big(z_-(v_-,\tau)+x_-\big)\right|d\tau d\sigma\nonumber\\
&&+\int_0^{+\infty}\int_\sigma^{+\infty}|F^s\big(z_-(v_-,\tau)+x_-+y_-(\tau)\big)-F^s\big(z_-(v_-,\tau)+x_-\big)|d\tau d\sigma\nonumber\\
&\le&{n\over (\alpha+1) \big({|v_-|\over 2\sqrt{2}}-r\big)^2(1+{|x_-|\over \sqrt{2}}-r)^{\alpha}}\Big({\beta_2\sup_{s\in (-\infty,0)}|y_-(s)|\over \alpha}
+{\beta_3^s\|y_-\|\over (\alpha+2)\big(1+{|x_-|\over \sqrt{2}}-r\big)}\nonumber\\
&&+{\beta_3^s\sup_{s\in (-\infty,0)}{|y_-(s)|\over 1+s}\over \alpha \big({|v_-|\over 2\sqrt{2}}-r\big)}\Big).\label{3.6}
\end{eqnarray}
Then $\|y_-\|$ is bounded by the right hand side of \eqref{l1}, and combining this latter bound on $\|y_-\|$  and \eqref{3.6} and the estimate $1+{|x_-|\over \sqrt{2}}-r\ge 1-r$ we obtain \eqref{504b}.

It remains to prove \eqref{503c}, \eqref{503d} and \eqref{503e}.
From \eqref{503a} it follows that
\begin{equation}
|a_{sc}(v_-,x_-)|\le {4n^{1\over 2}\max(\beta_1^l,\beta_2)\over\alpha\big({|v_-|\over 2 \sqrt{2}}-r\big)(1-r)^{\alpha+1}}.\label{3.15a}
\end{equation}
In addition using \eqref{5.103b} (for "$(w,v,x,h)=(a(v_-,x_-),a(v_-,x_-),0,0)$") and the identity $|a(v_-,x_-)|=|v_-|$ that follows from the conservation of energy, we have
\begin{equation}
|z_+(a(v_-,x_-),t)-t a(v_-,x_-)|\le {2^{5\over 2}n^{1\over 2}\beta_1^l\over \alpha|v_-|}|t|,\ t\in \R.\label{3.15aa}
\end{equation}
Then using \eqref{2.5a} for "$f=y_-$" and using \eqref{3.15a}, \eqref{3.15aa} and  \eqref{5.103b} (for "$(w,v,x,h)=(v_-,v_-,0,0)$") we obtain
\begin{eqnarray}
&&|x_-+\eta (z_-(v_-,t)+y_-(t))+(1-\eta)z_+(a(v_-,x_-),t)|\nonumber\\
&\ge & | x_-+t v_-|-|y_-(t)|-|z_-(v_-,t)-v_-t|-|z_+(a(v_-,x_-),t)-t a(v_-,x_-)|\nonumber\\
&&-|t||a_{sc}(v_-,x_-)|\nonumber\\
&\ge &{|x_-|\over \sqrt{2}}-r+ |t|\left( {|v_-|\over \sqrt{2}}-{2^{7\over 2}n^{1\over 2}\beta_1^l\over \alpha|v_-|}
-{4\max(\beta_1^l,\beta_2) n^{1\over 2}\over \alpha\big({|v_-|\over 2 \sqrt{2}}-r\big)(1-r)^{\alpha+1}}\right)\nonumber
\end{eqnarray}
\begin{equation}
\ge {|x_-|\over \sqrt{2}}-r+ |t|\left( {|v_-|\over \sqrt{2}}
-{8\max(\beta_1^l,\beta_2) n^{1\over 2}\over \alpha\big({|v_-|\over 2 \sqrt{2}}-r\big)(1-r)^{\alpha+1}}\right)\ge {|x_-|\over \sqrt{2}}-r+ |t|{|v_-|\over 2\sqrt{2}},\label{3.15}
\end{equation}
for $\eta\in (0,1)$ and $t\in\R$ (we used \eqref{503} and we used the estimate $| x_-+t v_-|\ge {|x_-|\over \sqrt{2}}+|t| {|v_-|\over \sqrt{2}}$ that follows from $x_-\cdot v_-=0$). Similarly
\begin{equation}
|\eta x_-+z_+(a(v_-,x_-),t)|\ge {|v_-|\over 2\sqrt{2}}|t|,\textrm{ for }(\eta,t)\in (0,1)\times \R.\label{3.16}
\end{equation}
From \eqref{502b}, \eqref{1.4a} and \eqref{3.16} it follows that
\begin{eqnarray*}
|l_1(v_-,x_-)|&\le &\beta_2 n|x_-|\int_0^{+\infty}\int_\sigma^{+\infty}\sup_{\eta\in (0,1)}(1+|\eta x_-+z_+(a(v_-,x_-),\tau)|)^{-\alpha-2}d\tau d\sigma\\
&\le&\beta_2 n|x_-|\int_0^{+\infty}\int_\sigma^{+\infty}
\big(1+{|v_-|\over 2\sqrt{2}}|\tau|\big)^{-\alpha-2}d\tau d\sigma,
\end{eqnarray*}
which gives \eqref{503c}.
Using \eqref{502c},  \eqref{1.4a} and \eqref{3.15} we obtain
\begin{equation}
|l_2(v_-,x_-,y_-)|\le n\beta_2\int_0^{+\infty}\int_\sigma^{+\infty}{|z_-(v_-,\tau)+y_-(\tau)-z_+(a(v_-,x_-),\tau)|d\tau d\sigma\over 
\big(1+{|x_-|\over \sqrt{2}}-r+ |\tau|{|v_-|\over 2\sqrt{2}}\big)^{\alpha+2}}.\label{701aa}
\end{equation}
Then using \eqref{300} and \eqref{300b} we have 
\begin{eqnarray}
&&|z_-(v_-,\tau)+y_-(\tau)-z_+(a(v_-,x_-),\tau)|\nonumber\\
&\le&|l(v_-,x_-,y_-)|+|l_1(v_-,x_-)|+|l_2(v_-,x_-,y_-)|+|y_+(\tau)|,\label{702aa}
\end{eqnarray}
for $\tau\in (0,+\infty)$. Combining \eqref{701aa} and \eqref{702aa} we obtain
\begin{eqnarray}
&&|l_2(v_-,x_-,y)|
\le \ep(v_-,x_-,0)\label{3.18a}\\
&&\times(|l_2(v_-,x_-,y_-)|+|l(v_-,x_-,y_-)|+|l_1(v_-,x_-)|+\sup_{(0,+\infty)}|y_+|),\nonumber
\end{eqnarray}
where
\begin{equation}
\ep(v_-,x_-,t):={n\beta_2\over \alpha(\alpha+1)\big({|v_-|\over 2^{3\over 2}}-r\big)^2(1+{|x_-|\over \sqrt{2}}-r+t\big({|v_-|\over 2^{3\over 2}}-r)\big)^\alpha},\label{3.19}
\end{equation}
for $t\ge 0$.
From \eqref{300c} and \eqref{2.7b} it follows that
\begin{eqnarray}
|y_+(t)|&\le& n^{-{1\over 2}}\ep(v_-,x_-,t)\label{3.19aa}\\
&&+\int_t^{+\infty}\int_\sigma^{+\infty}\big|F^l(z_-(v_-,\tau)+x_-+y_-(\tau))-F^l(z_+(a(v_-,x_-),\tau))\big|d\tau d\sigma,\nonumber
\end{eqnarray}
for $t\ge 0$.
Then similarly to \eqref{3.18a} we have 
\begin{eqnarray}
\sup_{(t,+\infty)}|y_+|
&\le&  \ep(v_-,x_-,t)(n^{-{1\over 2}}+|l_2(v_-,x_-,y_-)|+|l(v_-,x_-,y_-)|\nonumber\\
&&+|l_1(v_-,x_-)|+\sup_{(t,+\infty)}|y_+|),
\label{3.18b}
\end{eqnarray}
for $t\ge 0$.
From \eqref{3.18a} and \eqref{3.18b} it follows that
\begin{equation}
(1-\ep(v_-,x_-,0))|l_2(v_-,x_-,y_-)|\le\ep(v_-,x_-,0)(|l(v_-,x_-,y_-)|+|l_1(v_-,x_-)|+\sup_{(0,+\infty)}|y_+|),\label{3.20a}
\end{equation}
and
\begin{equation}
(1-\ep(v_-,x_-,t))\sup_{(t,+\infty)}|y_+|\le\ep(v_-,x_-,t) \big(n^{-{1\over 2}}+|l_2(v_-,x_-,y)|+|l(v_-,x_-,y_-)|+|l_1(v_-,x_-)|\big),\label{3.20b}
\end{equation}
for $t\ge 0$.
Using \eqref{503} and \eqref{3.19} we have 
\begin{equation}
\sup_{t\in(0,+\infty)}\ep(v_-,x_-,t)= \ep(v_-,x_-,0)\le 8^{-1}.\label{3.21a}
\end{equation}
Then multiplying \eqref{3.20a} by $(1-\ep(v_-,x_-,0))$ and using \eqref{3.20b} for $t=0$ we obtain 
\begin{eqnarray}
(1-2\ep(v_-,x_-,0))|l_2(v_-,x_-,y)|
&\le&\ep(v_-,x_-,0)^2n^{-{1\over 2}}\label{3.21}\\
&&+\ep(v_-,x_-,0)(|l(v_-,x_-,y)|+|l_1(v_-,x_-)|).\nonumber
\end{eqnarray}
Using \eqref{3.21} and \eqref{3.21a} we have
\begin{equation}
|l_2(v_-,x_-,y)|
\le2\ep(v_-,x_-,0)(n^{-{1\over 2}}\ep(v_-,x_-,0)+|l(v_-,x_-,y)|+|l_1(v_-,x_-)|).\label{3.21b}
\end{equation}  
Then \eqref{503d} follows from \eqref{3.19} (for $t=0$), \eqref{503b} and \eqref{2.26}.
Using \eqref{3.21a} and \eqref{3.20b} we obtain
\begin{equation}
\sup_{(t,+\infty)}|y_+|\le 2\ep(v_-,x_-,t) \big(n^{-{1\over 2}}+|l_2(v_-,x_-,y)|+|l(v_-,x_-,y_-)|+|l_1(v_-,x_-)|\big),
\end{equation}
for $t\ge 0$.
Then \eqref{503e} follows from \eqref{503d} (combined with \eqref{503}), \eqref{503b}, \eqref{503c} and \eqref{3.19}. \hfill $\Box$

\section{Proof of Lemma \ref{lem_cont3} and Theorem \ref{thm_y2}}
\label{proof_thm_y2}

\subsection{Preliminary lemmas}
\begin{lemma}
\label{lem_scatinit2}
Let $(v,x,w,h)\in (\R^n)^4$ and $x\in \R^n$  so that $v\cdot x=0$. Under the assumptions of Lemma \ref{lem_scatinit} the following estimates are valid
\begin{eqnarray}
|z_+(w,x+h,t)-z_+(w,x+h',t)|&\le& 2|h-h'|,\label{5.103c}\\
|\eta z_+(w,x+h',t)+(1-\eta)z_+(w,x+h,t)|&\ge &{|x|\over \sqrt{2}}-|h|+t{|v|\over 2\sqrt{2}},\label{5.103e}
\end{eqnarray}
for $t\ge 0$, $\eta\in(0,1)$ and $h'\in \R^n$, $|h'|\le |h|$.
\end{lemma}

\begin{proof}[Proof of Lemma \ref{lem_scatinit2}]
First we prove \eqref{5.103e}.
We estimate $\eta z_+(w,x+h',\tau)+(1-\eta)z_+(w,x+h,t)$ for $\eta \in (0,1)$ and $t\ge 0$, and for $(h,h')\in \B(0,1+{|x|\over \sqrt{2}})^2$ so that  $|h'|\le |h|$.
Using  \eqref{5.103b} we obtain
\begin{eqnarray}
&&|\eta z_+(w,x+h',t)+(1-\eta)z_+(w,x+h,t)|\nonumber\\
&\ge& |x+tv|-|\eta h'+(1-\eta)h|-t|v-w|-\eta|z_+(w,x+h',t)-x-h'-tw|\nonumber\\
&&-(1-\eta)|z_+(w,x+h,t)-x-h-tw|\nonumber
\end{eqnarray}
\begin{equation}
\ge {|x|\over \sqrt{2}}-|h|+t\Big({3|v|\over 4\sqrt{2}}-(1-\eta){2^{5\over 2}n^{1\over 2}\beta_1^l\over \alpha|v|(1+{|x|\over \sqrt{2}}-|h|)^\alpha}
-\eta{2^{5\over 2}n^{1\over 2}\beta_1^l\over \alpha|v|(1+{|x|\over \sqrt{2}}-|h'|)^\alpha}\Big)\label{5.110},
\end{equation}
for $t\ge 0$. Estimate \eqref{5.103e} follows from \eqref{5.110} and assumption \eqref{5.101}.

Now we set $\tilde \delta(t):=z_+(w,x+h',t)-z_+(w,x+h,t)$ for $t\ge 0$, where $|h'|\le |h|$.
From Lemma \ref{lem_comp} it follows that $\sup_{(0,+\infty)}|\tilde \delta|<\infty$.
Then from  \eqref{5.102} with the boundary conditions \eqref{5.103a} it follows that
\begin{equation}
\tilde \delta(t)
=h'-h-\int_0^t\int_\sigma^{+\infty}(F^l\big(z_+(w,x+h',\tau)\big)-F^l\big(z_+(w,x+h,\tau)\big))d\tau d\sigma,\label{5.117}
\end{equation}
for $t\ge 0$, and using \eqref{1.4a} and \eqref{5.103e} we obtain
\begin{equation}
|\tilde \delta(t)|\le |h'-h|+n\beta_2\sup_{(0,+\infty)}|\tilde \delta| \int_0^t\int_\sigma^{+\infty}
\Big(1+{|x|\over \sqrt{2}}-|h|+\tau{|v|\over 2\sqrt{2}}\Big)^{-\alpha-2}d\tau d\sigma,\label{5.126}
\end{equation}
for $t\ge 0$, which proves that 
\begin{equation}
\Big(1-{n\beta_2\over \alpha(\alpha+1)\big({|v|\over 2\sqrt{2}}\big)^2\big(1+{|x|\over \sqrt{2}}-|h|\big)^{\alpha}}\Big)\sup_{(0,+\infty)}|\tilde \delta|\le 
|h'-h|.\label{5.127}
\end{equation}
Using \eqref{5.101} we have $
1-{8n\beta_2\over \alpha(\alpha+1)|v|^2\big(1+{|x|\over \sqrt{2}}-|h|\big)^{\alpha}}\ge {1\over 2}$. Combining this latter estimate and \eqref{5.127} we obtain \eqref{5.103c}.
\end{proof}

We also need the following lemma.
\begin{lemma}
\label{lem:est}
Let $(v_-,x_-)\in \R^n\times\R^n$, $v_-\cdot x_-=0$ and let $r>0$, $r<{1\over 2}+{|x_-|\over 2^{3\over 2}}$, and let $y_-\in M_r$. Under the assumptions of Lemma \ref{lem_cont3}
the following estimates are valid
\begin{equation}
|z_-(v_-,x_-,t)+y_-(t)-z_+(\tilde a(v_-,x_-),x_-,t)|\le {{20\beta_2(nr+3\sqrt{n})\over 19\alpha(\alpha+1)\big({|v_-|\over 2\sqrt{2}}-r\big)^2
\big(1+{|x_-|\over \sqrt{2}}-r\big)^\alpha}},\label{5.5}
\end{equation}
\begin{equation}
|\eta\big(z_-(v_-,x_-,t)+y_-(t)\big)+(1-\eta)z_+(\tilde a(v_-,x_-),x_-,t)|
\ge {|x_-|\over \sqrt{2}}-r+|t|\big({|v_-|\over 2\sqrt{2}}-r\big),\label{5.5b}
\end{equation}
for $t\ge 0$ and $\eta\in(0,1)$.
\end{lemma}

\begin{proof}[Proof of Lemma \ref{lem:est}]
We set $\delta(t):=z_-(v_-,x_-,t)+y_-(t)-z_+(\tilde a(v_-,x_-),x_-,t)$ for $t\ge 0$.
From Lemma \ref{lem_comp} it follows that $\sup_{(0,+\infty)}|\delta|<\infty$.
Using \eqref{1.1} and \eqref{5.102} we obtain
\begin{eqnarray}
\delta(t)
&=&-\int_0^t\int_\sigma^{+\infty}(F^l\big(z_-(v_-,x_-,\tau)+y_-(\tau)\big)-F^l\big(z_+(\tilde a(v_-,x_-),x_-,\tau)\big))d\tau d\sigma\nonumber\\
&&+\tilde l(v_-,x_-,y_-)+H_1(v_-,x_-,y_-)(t),\label{5.6}
\end{eqnarray}
for $t\ge 0$, where $\tilde l$ is defined by \eqref{5.2b} and where
\begin{equation}
H_1(v_-,x_-,y_-)(t)=\int_t^{+\infty}\int_\sigma^{+\infty}F^s\big(z_-(v_-,x_-,\tau)+y_-(\tau)\big)d\tau d\sigma.\label{5.2c}
\end{equation}

Then from \eqref{5.2b}, \eqref{4.7b}, \eqref{4.7c} and the estimate $|y_-(\tau)|\le r$ for $\tau \le 0$ we obtain
\begin{equation}
|\tilde l(v_-,x_-,y_-)|\le{\beta_2(nr+2\sqrt{n})\over \alpha(\alpha+1)\big({|v_-|\over 2\sqrt{2}}-r\big)^2
\big(1+{|x_-|\over \sqrt{2}}-r\big)^\alpha}.\label{5.2d}
\end{equation}
Using \eqref{5.2c} and \eqref{4.7b} we obtain 
\begin{equation}
|H_1(v_-,x_-,y_-)(t)|\le{\beta_2\sqrt{n}\over \alpha(\alpha+1)\big({|v_-|\over 2\sqrt{2}}-r\big)^2\Big(1+{|x_-|\over \sqrt{2}}-r\Big)^{\alpha}},\label{5.2e}
\end{equation}
for $t\ge 0$.
From \eqref{5.9} and the inequality $r\le{1\over 2}+{|x_-|\over 2^{3\over 2}}$ and from  \eqref{5.2ba} it follows that
\begin{equation}
|\tilde a_{sc}(v_-,x_-)|\le {6\sqrt{n}\max(\beta_1^l,\beta_2)\over \alpha({|v_-|\over 2^{3\over 2}}-r)(1-r+{|x_-|\over \sqrt{2}})^\alpha}.\label{0.00}
\end{equation}
Then from \eqref{5.9a} it follows that $|a_{sc}(v_-,x_-)|\le{|v_-|\over 4\sqrt{2}}$, and using \eqref{5.103e} we obtain $|z_+(a(v_-,x_-),x_-,t)|\ge  {|x|\over \sqrt{2}}+t{|v|\over 2\sqrt{2}}$ for $t\ge 0$.
And using  \eqref{5.103b} we obtain
\begin{eqnarray}
&&|\eta\big(z_-(v_-,x_-,t)+y_-(t)\big)+(1-\eta)z_+(\tilde a(v_-,x_-),x_-,t)|\nonumber\\
&\ge& |x_-+tv_-|-(1-\eta)\big(|z_+(\tilde a(v_-,x_-),x_-,t)-x_--t\tilde a(v_-,x_-)|+t|\tilde a_{sc}(v_-,x_-)|\big)\nonumber\\
&&-\eta|z_-(v_-,x_-,t)-x_--v_-t|-\eta|y_-(t)|\nonumber
\end{eqnarray}
\begin{equation}
\ge {|x_-|\over \sqrt{2}}-r+|t|\Big({|v_-|\over \sqrt{2}}-{6\sqrt{n}\beta\over \alpha \big({|v_-|\over 2\sqrt{2}}-r\big)\big(1+{|x_-|\over \sqrt{2}}-r\big)^\alpha}
-{2^{7\over 2}n^{1\over 2}\beta_1^l\over \alpha|v_-|(1+{|x_-|\over \sqrt{2}})^\alpha}-r\Big),\label{5.11}
\end{equation}
for $t\ge 0$. Combining \eqref{5.11} and \eqref{5.9a} we obtain \eqref{5.5b}.

From \eqref{5.6}, \eqref{1.4a} and \eqref{5.11} it follows that
\begin{equation}
|\delta(t)|\le C(v_-,x_-,r)+n\beta_2\sup_{(0,+\infty)}|\delta| \int_0^t\int_\sigma^{+\infty}
\Big(1+{|x_-|\over \sqrt{2}}-r+\tau\big({|v_-|\over 2\sqrt{2}}-r\big)\Big)^{-\alpha-2}d\tau d\sigma,\label{5.16}
\end{equation}
for $t\ge 0$, where
\begin{equation}
C(v_-,x_-,r):={\beta_2(nr+3\sqrt{n})\over \alpha(\alpha+1)\big({|v_-|\over 2\sqrt{2}}-r\big)^2
\big(1+{|x_-|\over \sqrt{2}}-r\big)^\alpha}\ge |\tilde l(v_-,x_-,y_-)|+\sup_{(0,+\infty)}|H_1(v_-,x_-,y_-)|\label{5.16aa}
\end{equation} 
(see \eqref{5.2d} and \eqref{5.2e}).
From \eqref{5.16} it follows that
\begin{equation}
\Big(1-{n\beta_2\over \alpha(\alpha+1)\big({|v_-|\over 2\sqrt{2}}-r\big)^2\big(1+{|x_-|\over \sqrt{2}}-r\big)^{\alpha}}\Big)\sup_{(0,+\infty)}|\delta|\le 
C(v_-,x_-,r).\label{5.17}
\end{equation}
Using \eqref{5.9a} we have 
$1-{n\beta_2\over \alpha(\alpha+1)\big({|v_-|\over 2\sqrt{2}}-r\big)^2\big(1+{|x_-|\over \sqrt{2}}-r\big)^{\alpha}}\ge {19\over 20}$, 
and then estimate \eqref{5.5} follows from \eqref{5.17} and \eqref{5.16aa}. Hence Lemma \ref{lem:est} is proved.
\end{proof}

\subsection{Proof of Lemma \ref{lem_cont3}}
Let $(h,h')\in \overline{\B(0,{1\over 4}+{|x_-|\over 2^{5\over 2}})}$, $|h'|\le |h|$. Using \eqref{5.103c} and \eqref{5.5b} we have 
\begin{eqnarray}
&&|\eta\big(z_-(v_-,x_-,t)+y_-(t)\big)+(1-\eta)z_+(\tilde a(v_-,x_-),x_-+h,t)|\nonumber\\
&\ge& |\eta\big(z_-(v_-,x_-,t)+y_-(t)\big)+(1-\eta)z_+(\tilde a(v_-,x_-),x_-,t)|\nonumber\\
&&-(1-\eta)|z_+(\tilde a(v_-,x_-),x_-,t)-z_+(\tilde a(v_-,x_-),x_-+h,t)|\nonumber
\end{eqnarray}
\begin{equation}
\ge {|x_-|\over \sqrt{2}}-r-2|h|+|t|\big({|v_-|\over 2\sqrt{2}}-r\big)\ge{|x_-|\over 2\sqrt{2}}-{1\over 2}-r+|t|\big({|v_-|\over 2\sqrt{2}}-r\big),\label{7.5}
\end{equation}
for $t\ge 0$ and $\eta\in(0,1)$.
Using \eqref{5.103c} and \eqref{5.5} we obtain 
\begin{equation}
|z_-(v_-,x_-,t)+y_-(t)-z_+(\tilde a(v_-,x_-),x_-+h,t)|\le {20\beta_2(nr+3\sqrt{n})\over 19\alpha(\alpha+1)\big({|v_-|\over 2\sqrt{2}}-r\big)^2
\big(1+{|x_-|\over \sqrt{2}}-r\big)^\alpha}+2|h|,\label{7.6}
\end{equation}
Combining \eqref{5.2d}, \eqref{7.4}, \eqref{1.4a}, \eqref{7.5} and \eqref{7.6} we have
\begin{eqnarray}
|\G_{v_-,x_-}(h)|&\le&{\beta_2(nr+2\sqrt{n})\over \alpha(\alpha+1)\big({|v_-|\over 2\sqrt{2}}-r\big)^2
\big(1+{|x_-|\over \sqrt{2}}-r\big)^\alpha}\nonumber\\
&&+n\beta_2\int_0^{+\infty}\int_\sigma^{+\infty} { {20\beta_2(nr+3\sqrt{n})\over 19\alpha(\alpha+1)\big({|v_-|\over 2\sqrt{2}}-r\big)^2
\big(1+{|x_-|\over \sqrt{2}}-r\big)^\alpha}+2|h|\over \big({1\over 2}+{|x_-|\over 2\sqrt{2}}-r+\tau\big({|v_-|\over 2\sqrt{2}}-r\big)\big)^{\alpha+2}}d\tau d\sigma\nonumber\\
&\le&{\beta_2(3(nr+\sqrt{n})+2|h|n)\over \alpha(\alpha+1)\big({|v_-|\over 2\sqrt{2}}-r\big)^2
\big({1\over 2}+{|x_-|\over 2\sqrt{2}}-r\big)^\alpha}\label{7.7}
\end{eqnarray}
(we used \eqref{5.9a} which implies ${60n\beta_2\over 19 \alpha(\alpha+1) \big({|v_-|\over 2\sqrt{2}}-r\big)^2\big(1+{|x_-|\over \sqrt{2}}-r\big)^\alpha}\le 1$). Then
using \eqref{7.7} and the estimate $|h|\le {1\over4}+{|x_-|\over 4\sqrt{2}}$ we obtain \eqref{7.8a}.
Using again \eqref{5.9a} and \eqref{7.8a} and the estimate $r\le{1\over 2}+{|x_-|\over 2^{3\over2}}$ we obtain \eqref{7.8b}.

Then using \eqref{7.4}, \eqref{1.4a}, \eqref{5.103e} and \eqref{5.103c} we have
\begin{eqnarray}
|\G_{v_-,x_-}(h)-\G_{v_-,x_-}(h')|&\le&\int_0^{+\infty}\int_\sigma^{+\infty}\big|F^l\big(z_+(\tilde a(v_-,x_-),x_-+h',\tau)\big)\nonumber\\
&&-F^l\big(z_+(\tilde a(v_-,x_-),x_-+h,\tau)\big)\big|d\tau d\sigma\nonumber\\
&\le& n\beta_2^l\int_0^{+\infty}\int_\sigma^{+\infty}{2|h-h'|d\tau d\sigma\over \big({1\over 2}+{|x_-|\over 2\sqrt{2}}+\tau{|v_-|\over 2\sqrt{2}}\big)^{\alpha+2}},
\end{eqnarray}
for $(h,h')\in \B(0, {1\over4}+{|x_-|\over 4\sqrt{2}})^2$,
which proves the first estimate in \eqref{7.9}. The second estimate in \eqref{7.9} follows from  \eqref{5.9a}.\hfill $\Box$

\subsection{Proof of Theorem \ref{thm_y2}}
The estimate \eqref{7.2c} (resp. \eqref{7.2d}) follows from \eqref{4.10} (resp. \eqref{4.11}) and the equality $y_-=\A(y_-)$.
For estimate  \eqref{7.14b} see also \eqref{0.00}.
We prove \eqref{7.14}.
Since $\tilde b_{sc}(v_-,x_-)=\G_{v_-,x_-}(\tilde b_{sc}(v_-,x_-))$ we use \eqref{7.7} with "$h=\tilde b_{sc}(v_-,x_-)$" and we obtain
\begin{eqnarray}
&&\left(1-{2n\beta_2\over \alpha(\alpha+1) \big({|v_-|\over 2\sqrt{2}}-r\big)^2\big({1\over2}+{|x_-|\over 2\sqrt{2}}-r\big)^\alpha}\right)|\tilde b_{sc}(v_-,x_-)|\nonumber\\
&\le&{3\beta_2(nr+\sqrt{n})\over \alpha(\alpha+1)\big({|v_-|\over 2\sqrt{2}}-r\big)^2
\big({1\over 2}+{|x_-|\over 2\sqrt{2}}-r\big)^\alpha}
.\label{7.13}
\end{eqnarray}
Then use \eqref{5.9a} to prove that $1-{2n\beta_2\over \alpha(\alpha+1) \big({|v_-|\over 2\sqrt{2}}-r\big)^2\big({1\over2}+{|x_-|\over 2\sqrt{2}}-r\big)^\alpha}\ge{9\over 10}$. This proves \eqref{7.14}.
We prove \eqref{7.16b}.
From  \eqref{1.4a}, \eqref{7.5} and \eqref{7.2a} it follows that
\begin{eqnarray}
&&\int_t^{+\infty}\int_\sigma^{+\infty}|F^l(z_-(v_-,x_-,\tau)+y_-(\tau))-F^l(z_+(\tilde a(v_-,x_-),\tilde b(v_-,x_-),\tau))|d\tau d\sigma\nonumber\\
&\le&\beta_2 n\int_t^{+\infty}\int_\sigma^{+\infty}{|y_+(\tau)| d\tau d\sigma
\over \big(1+{|x_-|\over \sqrt{2}}-r+\tau\big({|v_-|\over 2\sqrt{2}}-r\big)\big)^{\alpha+2}}\nonumber\\
&\le&{\beta_2 n\sup_{(t,+\infty)}|y_+|
\over \alpha(\alpha+1)\big({|v_-|\over 2\sqrt{2}}-r\big)^2 \big(1+{|x_-|\over \sqrt{2}}-r\big)^\alpha},\label{701}
\end{eqnarray}
for $t\ge 0$.
Then using \eqref{7.2b}, \eqref{701} and \eqref{5.2e} we obtain
\begin{equation}
|y_+(t)|\le{\beta_2(n^{1\over 2}+n\sup_{(t,+\infty)}|y_+|)\over \alpha(\alpha+1)\big({|v_-|\over 2\sqrt{2}}-r\big)^2\Big(1+{|x_-|\over \sqrt{2}}-r\Big)^{\alpha}},\label{702}
\end{equation}
for $t\ge0$. Then from \eqref{5.9a} it follows that ${\beta_2n\over \alpha(\alpha+1)\big({|v_-|\over 2\sqrt{2}}-r\big)^2\Big(1+{|x_-|\over \sqrt{2}}-r\Big)^{\alpha}}\le{1\over 20}$. Hence from \eqref{702} it follows that 
\begin{equation}
\sup_{(t,+\infty)}|y_+|\le{20\beta_2 n^{1\over 2}\over 19\alpha(\alpha+1)\big({|v_-|\over 2\sqrt{2}}-r\big)^2\Big(1+{|x_-|\over \sqrt{2}}-r\Big)^{\alpha}},\label{701b}
\end{equation}
which proves \eqref{7.16b}.

We prove \eqref{7.14c}.
From the contraction estimate on $\dot \A(.)(t)$ for $t\le 0$ (see the proof of Lemma \ref{lem_cont2} and the right hand side of \eqref{2.21} for "$(f_1,f_2,t)=(y_-,0,0)$") it follows that
\begin{eqnarray}
&&\left|\int_{-\infty}^0\big(F\big(z_-(v_-,x_-,\tau)+y_-(\tau)\big)-F\big(z_-(v_-,x_-,\tau)\big)\big)d\tau\right|\nonumber\\
&&\le {n\max(\beta_2,\beta_3^s)\sup_{(-\infty,0]}|y_-|\over (\alpha+1)\Big({|v_-|\over 2\sqrt{2}}-r\Big)\Big(1+{|x_-|\over \sqrt{2}}-r\Big)^{\alpha+1}}\big(1
+{1\over \big(1+{|x_-|\over \sqrt{2}}-r\big)}\big).
\label{5.3}
\end{eqnarray}
Using \eqref{4.20b} for "$(f_1,f_2)=(y_-,0)$" we also obtain
\begin{eqnarray}
&&\left|\int_0^{+\infty}\big(F^s\big(z_-(v_-,x_-,\tau)+y_-(\tau)\big)-F^s(z_-(v_-,x_-,\tau))\big)d\tau\right|\nonumber\\
&\le& {n\beta_3^s\sup_{s\in [0,+\infty)}{|y_-(s)|\over 1+s}\over(\alpha+1)\Big({|v_-|\over 2\sqrt{2}}-r\Big)(1+{|x_-|\over \sqrt{2}}-r)^{\alpha+1}}\left[{1\over 
(1+{|x_-|\over \sqrt{2}}-r)}\right.\nonumber\\
&&+\left.{1\over\Big({|v_-|\over 2\sqrt{2}}-r\Big)}\right].\label{5.21}
\end{eqnarray}
Then using \eqref{1.4a}, \eqref{5.5} and \eqref{5.5b} we have
\begin{eqnarray}
&&\left|\int_0^{+\infty}\big(F^l\big(z_-(v_-,x_-,\tau)+y_-(\tau)\big)-F^l\big(z_+(\tilde a(v_-,x_-),x_-,\tau)\big)\big)d\tau\right|\nonumber\\
&\le& {20n\beta_2^2(nr+3\sqrt{n})\over 19\alpha(\alpha+1)^2\big({|v_-|\over 2\sqrt{2}}-r\big)^3
\big(1+{|x_-|\over \sqrt{2}}-r\big)^{2\alpha+1}}.\label{5.20}
\end{eqnarray}
Combining \eqref{5.2ba}, \eqref{5.3}, \eqref{5.21} and \eqref{5.20}, we obtain 
\begin{eqnarray}
&&\big|\tilde a_{sc}(v_-,x_-)-\int_{-\infty}^0F\big(z_-(v_-,x_-,\tau)\big)d\tau\nonumber\\
&&-\int_0^{+\infty}F^s(z_-(v_-,x_-,\tau))d\tau
-\int_0^{+\infty}F^l\big(z_+(\tilde a(v_-,x_-),x_-,\tau)\big)d\tau|\nonumber\\
&\le&{n\max(\beta_2,\beta_3^s)\|y_-\|\over (\alpha+1)\Big({|v_-|\over 2\sqrt{2}}-r\Big)\Big(1+{|x_-|\over \sqrt{2}}-r\Big)^{\alpha+1}}\big(1
+{1\over\big(1+{|x_-|\over \sqrt{2}}-r\big)}+{1\over\big({|v_-|\over 2\sqrt{2}}-r\big)}\big)\nonumber\\
&&+{20n\beta_2^2(nr+3\sqrt{n})\over 19\alpha(\alpha+1)^2\big({|v_-|\over 2\sqrt{2}}-r\big)^3
\big(1+{|x_-|\over \sqrt{2}}-r\big)^{2\alpha+1}}.\label{800}
\end{eqnarray}
Then estimate \eqref{7.14c} follows from \eqref{800}, \eqref{lb1} and the equality $y_-=\A(y_-)$, and from the estimate $1+{|x_-|\over \sqrt{2}}-r\ge {1\over 2}$. 
We prove \eqref{6.44}.
From \eqref{5.2b} it follows that
\begin{eqnarray}
&&|\tilde l(v_-,x_-,y_-)-\tilde l(v_-,x_-,0)|\nonumber\\
&\le&\int_{-\infty}^0\int_{-\infty}^\sigma \big|F\big(z_-(v_-,x_-,\tau)+y_-(\tau)\big)
-F\big(z_-(v_-,x_-,\tau)\big)\big|d\tau d\sigma\label{6.1}\\
&&+\int_0^{+\infty}\int_\sigma^{+\infty}\Big|F^s\big(z_-(v_-,x_-,\tau)+y_-(\tau)\big)-F^s\big(z_-(v_-,x_-,\tau)\big)\Big|d\tau d\sigma\nonumber
\end{eqnarray}
From the contraction estimate on $\A(.)(t)$ for $t\le 0$ (see the proof of Lemma \ref{lem_cont2} and the right hand side of \eqref{2.22} for "$(f_1,f_2,t)=(y_-,0,0)$")
we have
$$
\left|\int_{-\infty}^0\int_{-\infty}^\sigma \left(F\big(z_-(v_-,x_-,\tau)+y_-(\tau)\big)
-F\big(z_-(v_-,x_-,\tau)\big)\right)d\tau d\sigma\right|
$$
\begin{equation}
\le{n\sup_{(-\infty,0]}|y_-|\over(\alpha+1) \big({|v_-|\over 2\sqrt{2}}-r\big)^2\big(1+{|x_-|\over \sqrt{2}}-r\big)^\alpha}\left[ {\beta_2\over \alpha}
+ {\beta_3^s\over (\alpha+2)\big(1+{|x_-|\over \sqrt{2}}-r\big)}\right],
\label{6.2}
\end{equation}
and using \eqref{4.20b} for "$(f_1,f_2)=(y_-,0)$" we obtain 
\begin{eqnarray}
&&\int_0^{+\infty}\int_\sigma^{+\infty}\Big|F^s\big(z_-(v_-,x_-,\tau)+y_-(\tau)\big)-F^s\big(z_-(v_-,x_-,\tau)\big)\Big|d\tau d\sigma\nonumber\\
&\le&n\beta_3^s\sup_{s\in (0,+\infty)}{|y_-(s)|\over 1+s}
\int_0^{+\infty}\int_\sigma^{+\infty}{(1+\tau)d\tau d\sigma\over \big(1+{|x_-|\over \sqrt{2}}-r+\tau\big({|v_-|\over 2\sqrt{2}}-r\big) \big)^{\alpha+3}}\nonumber
\end{eqnarray}
\begin{equation}
\le{n\beta_3^s\sup_{s\in (0,+\infty)}{|y_-(s)|\over 1+s}
\over(\alpha+1)\big({|v_-|\over 2\sqrt{2}}-r\big)^2 \big(1+{|x_-|\over \sqrt{2}}-r\big)^{\alpha}}
\left[{1\over \alpha\big({|v_-|\over 2\sqrt{2}}-r\big)}+{1\over (\alpha+2) \big(1+{|x_-|\over \sqrt{2}}-r\big)}\right].\label{6.3}
\end{equation}
Hence 
\begin{eqnarray}
&&|\tilde l(v_-,x_-,y_-)-\tilde l(v_-,x_-,0)|
\le{n\max(\beta_2,\beta_3^s)\|y_-\|\over\alpha(\alpha+1) \big({|v_-|\over 2\sqrt{2}}-r\big)^2\big(1+{|x_-|\over \sqrt{2}}-r\big)^\alpha}\nonumber\\
&&\times \big(1+{1\over\big(1+{|x_-|\over \sqrt{2}}-r\big)}+{1\over \big({|v_-|\over 2\sqrt{2}}-r\big)}\big).\label{6.4}
\end{eqnarray}
And combining \eqref{701} and \eqref{701b} we have
\begin{eqnarray}
&&\int_0^{+\infty}\int_\sigma^{+\infty}|F^l(z_-(v_-,x_-,\tau)+y_-(\tau))-F^l(z_+(\tilde a(v_-,x_-),\tilde b(v_-,x_-),\tau))|d\tau d\sigma\nonumber\\
&\le&{20\beta_2^2 n^{3\over 2}
\over 19\alpha^2(\alpha+1)^2\big({|v_-|\over 2\sqrt{2}}-r\big)^4 \big(1+{|x_-|\over \sqrt{2}}-r\big)^{2\alpha}}.\label{703}
\end{eqnarray}
for $t\ge 0$.
Then we use the estimates \eqref{6.4} and \eqref{lb1} and the equality $y_-=\A(y_-)$, and we use the  estimate \eqref{703} and the estimate $1+{|x_-|\over \sqrt{2}}-r\ge {1\over 2}$ to obtain \eqref{6.44}.\hfill $\Box$

\section*{Acknowledgments}
The author would like to thank R. G. Novikov for comments on a previous version of this paper.

\vskip 8mm

\noindent A. Jollivet

\noindent Laboratoire de Physique Th\'eorique et Mod\'elisation,

\noindent CNRS UMR 8089/Universit\'e de Cergy-Pontoise

\noindent 95302 Cergy-Pontoise, France

\noindent e-mail: alexandre.jollivet@u-cergy.fr

\end{document}